\documentclass[12pt,onecolumn,journal,draftclsnofoot,doublespace]{IEEEtran}
\usepackage{epsf,psfrag,amssymb,amsfonts,amsmath,graphicx,subfigure}
\usepackage[square,comma,sort&compress]{natbib}
\usepackage[mathscr]{eucal}

\def\boxit#1{\vbox{\hrule\hbox{\vrule\kern3pt
        \vbox{\kern3pt#1\kern3pt}\kern3pt\vrule}\hrule}}

\def\reals{ { {\rm  I \kern-0.15em R }  } }
\def\complex{ {\,{{\rm C} \kern-0.50em \raise0.20ex {  |}}\, }}

\def\Rbf{{\bf R}}

\def\Cc{{\cal C}}

\def\Gc{{\cal G}}

\def\Lc{{\cal L}}

\def\Sc{{\cal S}}

\def\be{\begin{equation}}
\def\ee{\end{equation}}

\def\defeq{{\stackrel{\Delta}{=}}}
\def\scalefig#1{\epsfxsize #1\textwidth}
\def\Rxx{\Rbf_{\ssstyle X\kern-.1em X}}

\let\ssstyle=\scriptscriptstyle

\def\etal{{\it et al. \/}}

\def\ie{{\it i.e.,\ \/}}
\def\Kout{\setbox1=\hbox{\Huge\bf K}\hbox to
1.05\wd1{\hspace{.05\wd1}
}}
\def\Sout{\setbox1=\hbox{\Huge\bf S}\hbox to 1.05\wd1{\hspace{.05\wd1}
}}

\def\ie{{\it i.e.,\ \/}}

\def\defeq{{\,\stackrel{\Delta}{=}\,}}

\def\scalefig#1{\epsfxsize #1\textwidth}
\def\nn{{\nonumber}}

\newtheorem{lemma}{Lemma}
\newtheorem{theorem}{Theorem}
\newtheorem{proposition}{Proposition}

\newtheorem{corollary}{Corollary}
\newtheorem{fact}{Fact}

\begin{document}
\title{\LARGE Connectivity of Heterogeneous Wireless Networks\thanks{This work was supported in part by
the Army Research Laboratory CTA on Communication and Networks under
Grant DAAD19-01-2-0011, by the Army Research Office under Grant
W911NF-08-1-0467, and by the National Science Foundation under
Grants ECS-0622200 and CCF-0830685.}}
\author{Wei Ren, ~~~Qing Zhao$^*$, ~~~Ananthram Swami
\thanks{W. Ren and Q. Zhao are with the
Department of Electrical and Computer Engineering, University of
California, Davis, CA 95616. A. Swami is with the Army Research
Laboratory, Adelphi, MD 20783.}
\thanks{$^*$ Corresponding author. Phone: 1-530-752-7390. Fax: 1-530-752-8428. Email: qzhao@ece.ucdavis.edu.}}

\markboth{Submitted to {\it IEEE Transactions on Information
Theory}, August, 2009.}{Ren, Zhao, and Swami}

\maketitle \vspace{-1em}

\begin{abstract}
We address the connectivity of large-scale ad hoc heterogeneous
wireless networks, where secondary users exploit channels
temporarily unused by primary users and the existence of a
communication link between two secondary users depends on not only
the distance between them but also the transmitting and receiving
activities of nearby primary users. We introduce the concept of {\em
connectivity region} defined as the set of density pairs --- the
density of secondary users and the density of primary transmitters
--- under which the secondary network is connected. Using theories
and techniques from continuum percolation, we analytically
characterize the connectivity region of the secondary network and
reveal the tradeoff between proximity (the number of neighbors) and
the occurrence of spectrum opportunities. Specifically, we establish
three basic properties of the connectivity region -- contiguity,
monotonicity of the boundary, and uniqueness of the infinite
connected component, where the uniqueness implies the occurrence of
a phase transition phenomenon in terms of the almost sure existence
of either zero or one infinite connected component; we identify and
analyze two critical densities which jointly specify the profile as
well as an outer bound on the connectivity region; we study the
impacts of secondary users' transmission power on the connectivity
region and the conditional average degree of a secondary user, and
demonstrate that matching the interference ranges of the primary and
the secondary networks maximizes the tolerance of the secondary
network to the primary traffic load. Furthermore, we establish a
necessary condition and a sufficient condition for connectivity,
which lead to an outer bound and an inner bound on the connectivity
region.
\end{abstract}



\section{Introduction}
\label{sec:intro}

The communication infrastructure is becoming increasingly
heterogeneous, with a dynamic composition of interdependent,
interactive, and hierarchical network components with different
priorities and service requirements. One example is the cognitive
radio technology~\citep{Mitola&Maguire:99} for opportunistic
spectrum access which adopts a hierarchical structure for resource
sharing~\citep{Zhao&Sadler:07SPM}. Specifically, a secondary network
is overlaid with a primary network, where secondary users identify
and exploit temporarily and locally unused channels without causing
unacceptable interference to primary
users~\citep{Zhao&Sadler:07SPM}.

\subsection{Connectivity and Connectivity Region}
\label{subsec:CCR}

While the connectivity of homogeneous ad hoc networks consisting of
peer users has been well studied (see, for example,
~\citep{Dousse&Etal:05ITN,Dousse&Etal:06JAPPP,Gupta&Kumar:98SACOA,
Philips&Etal:89ITIT,Ni&Chandler:94IEEPC,Bettstetter:02MobiHoc,Kong&Yeh:08Infocom,Kong&Yeh:09ITIT}),
little is known about the connectivity of heterogeneous networks.
The problem is fundamentally different from its counterpart in
homogeneous networks. In particular, the connectivity of the
low-priority network component depends on the characteristics
(traffic pattern/load, topology, interference tolerance, etc.) of
the high-priority component, thus creating a much more diverse and
complex design space.

Using theories and techniques from continuum percolation, we
analytically characterize the connectivity of the secondary network
in a large-scale ad hoc heterogeneous network. Specifically, we
consider a Poisson distributed secondary network overlaid with a
Poisson distributed primary network in an infinite two-dimensional
Euclidean space\footnote{This infinite network model is equivalent
in distribution to the limit of a sequence of finite networks with a
fixed density as the area of the network increases to infinity, \ie
the so-called \emph{extended
network}~\citep{Leveque&Telatar:05ITIT}. It follows from the
arguments similar to the ones used in~\citep[Chapter
3]{Franceschetti&Meester:Random_Network_Comm} for homogeneous ad hoc
networks that this infinite ad hoc heterogeneous network model
represents the limiting behavior of large-scale networks.}. We
define {\em network connectivity} as the existence of an infinite
connected component almost surely (a.s.), \ie the occurrence of
percolation. We say that the secondary network is strongly connected
when it contains a \emph{unique} infinite connected component a.s.

Due to the hierarchical structure of spectrum sharing, a
communication link exists between two secondary users if the
following two conditions hold: (C1) they are within each other's
transmission range; (C2) they see a spectrum opportunity determined
by the transmitting and receiving activities of nearby primary users
(see Sec.~\ref{subsubsec:CL}). Thus, given the transmission power
and the interference tolerance of both the primary and the secondary
users, the connectivity of the secondary network depends on the
density of secondary users (due to (C1)) and the traffic load of
primary users (due to (C2)).

We thus introduce the concept of \emph{connectivity region} $\Cc$,
defined as the set of density pairs $(\lambda_S,\lambda_{PT})$ under
which the secondary network is connected, where $\lambda_S$ denotes
the density of the secondary users and $\lambda_{PT}$ the density of
primary transmitters (representing the traffic load of the primary
users). As illustrated in Fig.~\ref{fig:Con_Region}, a secondary
network with a density pair $(\lambda_S,\lambda_{PT})$ inside this
region is connected: the secondary network has a giant connected
component which includes infinite secondary users. The existence of
the giant connected component enables bidirectional communications
between distant secondary users via multihop relaying. On the other
hand, a secondary network with a density pair
$(\lambda_S,\lambda_{PT})$ outside this region is not connected: the
network is separated into an infinite number of \emph{finite}
connected components. Consequently, any secondary user can only
communicate with users within a limited range.

\begin{figure}[htbp]
\centerline{
\begin{psfrags}
\psfrag{lP}[c]{$\lambda_{PT}$} \psfrag{lS}[c]{$\lambda_S$}
\psfrag{lPu}[c]{$\overline{\lambda^*_{PT}}$}
\psfrag{lPlS}[c]{$\lambda^*_{PT}(\lambda_S)$}
\psfrag{lSc}[c]{$\lambda^*_S$} \psfrag{O}[c]{$0$}
\psfrag{CR}[c]{\textbf{Connectivity Region}}
\scalefig{0.7}\epsfbox{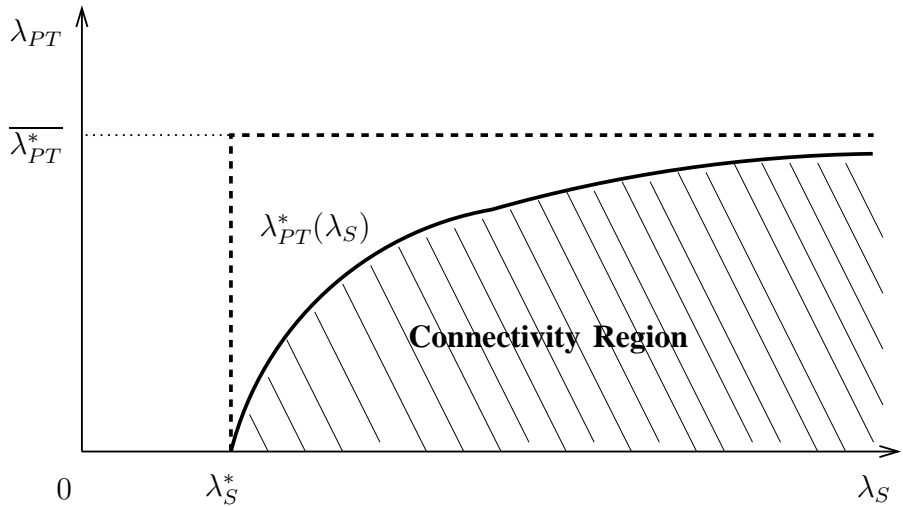}
\end{psfrags}
} \caption{The connectivity region $\mathcal{C}$ (the upper boundary
$\lambda^*_{PT}(\lambda_S)$ is defined as the supremum density of
the primary transmitters to ensure connectivity with a \emph{fixed}
density of the secondary users; the critical density $\lambda^*_S$
of the secondary users is defined as the infimum density of the
secondary users to ensure connectivity under a \emph{positive}
density of the primary transmitters; the critical density
$\overline{\lambda^*_{PT}}$ of the primary transmitters the supremum
density of the primary transmitters to ensure connectivity with a
\emph{finite} density of the secondary users).}
\label{fig:Con_Region}
\end{figure}

The objective of this paper is to establish analytical
characterizations of the connectivity region and to study the impact
of system design parameters (in particular, the transmission power
of the secondary users) on the network connectivity. Main results
are summarized in the subsequent two subsections.

\subsection{Analytical Characterizations of the Connectivity Region}
\label{subsec:intro_AC_CR}

We first establish three basic properties of the connectivity
region: contiguity, monotonicity of the boundary, and uniqueness of
the infinite connected component. Specifically, based on a coupling
argument, we show that the connectivity region is a contiguous area
bounded below by the $\lambda_S$-axis and bounded above by a
monotonically increasing function $\lambda_{PT}^*(\lambda_S)$ (see
Fig.~\ref{fig:Con_Region}), where the upper boundary
$\lambda_{PT}^*(\lambda_S)$ is defined as
\[
\lambda^*_{PT}(\lambda_S)\defeq
\sup\{\lambda_{PT}:~\mathcal{G}(\lambda_S,\lambda_{PT}) \mbox{ is
connected.}\},
\]
with $\mathcal{G}(\lambda_S,\lambda_{PT})$ denoting the secondary
network of density $\lambda_S$ overlaid with a primary network
specified by the density $\lambda_{PT}$ of the primary transmitters.
The uniqueness of the infinite connected component is established
based on the ergodic theory and certain combinatorial results. It
shows that once the secondary network is connected, it is strongly
connected.

Second, we identify and analyze two critical parameters of the
connectivity region: $\lambda^*_S$ and $\overline{\lambda^*_{PT}}$.
They jointly specify the profile as well as an outer bound on the
connectivity region. Referred to as the critical density of the
secondary users, $\lambda^*_S$ is the infimum density of the
secondary users to ensure connectivity under a positive density of
the primary transmitters:
\begin{equation*}
\lambda^*_S\defeq \inf\{\lambda_S: \exists \lambda_{PT}>0 \mbox{
s.t. $\mathcal{G}(\lambda_S,\lambda_{PT})$ is connected}\}.
\end{equation*}
We show that $\lambda^*_S$ equals the critical density $\lambda_c$
of a \emph{homogeneous} ad hoc network (\ie in the absence of
primary users), which has been well
studied~\citep{Meester&Roy:Con_Percolation}. This result shows that
the ``takeoff'' point in the connectivity region is completely
determined by the effect of proximity---the number of neighbors
(nodes within the transmission range of a secondary user).

Referred to as the critical density of the primary transmitters,
$\overline{\lambda^*_{PT}}$ is the supremum density of the primary
transmitters to ensure the connectivity of the secondary network
with a finite density of the secondary users:
\begin{equation*}
\overline{\lambda^*_{PT}}\defeq \sup\{\lambda_{PT}: \exists
\lambda_S<\infty \mbox{ s.t. $\mathcal{G}(\lambda_S,\lambda_{PT})$
is connected}\}.
\end{equation*}
We obtain an upper bound on $\overline{\lambda^*_{PT}}$ which is
shown to be achievable in simulations. More importantly, this result
shows that when the density of the primary transmitters is higher
than the (finite) value given by this upper bound, the secondary
network cannot be connected no matter how dense it is. This
parameter $\overline{\lambda^*_{PT}}$ thus characterizes the impact
of opportunity occurrence on the connectivity of the secondary
network: when the density of the primary transmitters is beyond a
certain level, there are simply not enough spectrum opportunities
for any secondary network to be connected.

Since a precise characterization of the upper boundary
$\lambda^*_{PT}(\lambda_S)$ of the connectivity region is
intractable, we establish a necessary and a sufficient condition for
connectivity to provide an outer and an inner bound on the
connectivity region. The necessary condition is expressed in the
form of the conditional average degree of a secondary user, and is
derived by the construction of a branching process. The sufficient
condition is obtained by the discretization of the continuum
percolation model into a dependent site percolation model.

\begin{figure}[htbp]
\centerline{
\scalefig{0.7}\epsfbox{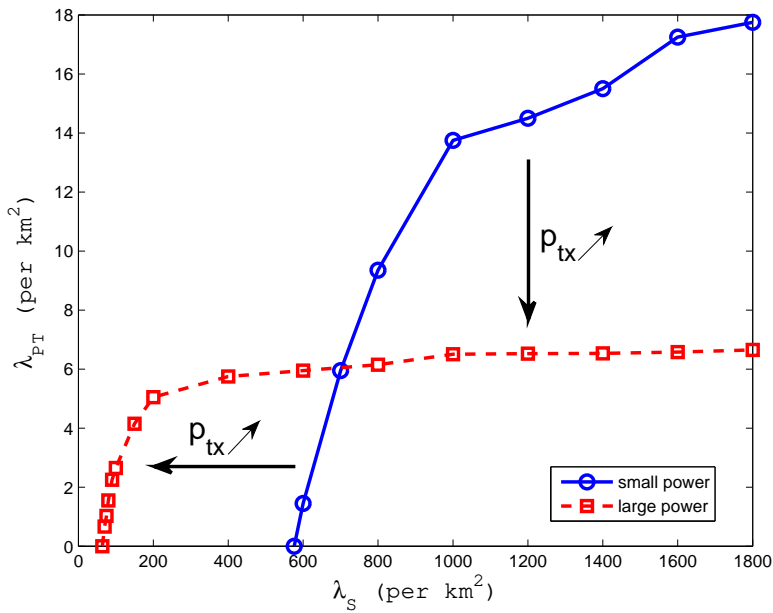} }
\caption{Simulated connectivity regions for two different
transmission powers ($p_{tx}$ denotes the transmission power of the
secondary users, and the large $p_{tx}$ is $3^{\alpha}$ times the
small $p_{tx}$, where $\alpha$ is the path-loss exponent).}
\label{fig:Con_Region_Simu1}
\end{figure}

\subsection{Impact of Transmission Power on Connectivity: Proximity vs. Opportunity}
\label{subsec:intro_PvsO}

The study on the impact of the secondary users' transmission power
on the network connectivity reveals an interesting tradeoff between
proximity and opportunity in the design of heterogeneous networks.
As illustrated in Fig.~\ref{fig:Con_Region_Simu1}, we show that
increasing $p_{tx}$ enlarges the connectivity region $\Cc$ in the
$\lambda_S$-axis (\ie better proximity leads to a smaller
``takeoff'' point), but at the price of reducing $\Cc$ in the
$\lambda_{PT}$-axis. Specifically, with a large $p_{tx}$, few
secondary users experience spectrum opportunities due to their large
interference range with respect to the primary users. This leads to
a poor tolerance to the primary traffic load parameterized by
$\lambda_{PT}$.

The transmission power $p_{tx}$ of the secondary network should thus
be chosen according to the operating point of the heterogeneous
network given by the density of the secondary users and the traffic
load of the co-existing primary users. Using the tolerance to the
primary traffic load as the performance measure, we show that the
interference range $r_I$ of the secondary users should be equal to
the interference range $R_I$ of the primary users in order to
maximize the upper bound on the critical density
$\overline{\lambda^*_{PT}}$ of the primary transmitters. Given the
interference tolerance of the primary and secondary users, we can
then design the optimal transmission power $p_{tx}$ of the secondary
users based on that of the primary users.

\subsection{Related Work}
\label{subsec:RW}

To our best knowledge, the connectivity of large-scale ad hoc
heterogeneous networks has not been characterized analytically or
experimentally in the literature. There are a number of classic
results on the connectivity of homogeneous ad hoc networks. For
example, it has been shown that to ensure either $1$-connectivity
(there exists a path between any pair of
nodes)~\citep{Gupta&Kumar:98SACOA, Philips&Etal:89ITIT} or
$k$-connectivity (there exist at least $k$ node-disjoint paths
between any pair of nodes)~\citep{Bettstetter:02MobiHoc}, the
average number of neighbors of each node must increase with the
network size. On the other hand, to maintain a weaker connectivity
-- p-connectivity (\ie the probability that any pair of nodes is
connected is at least $p$), the average number of neighbors is only
required to be above a certain `magic number' which does not depend
on the network size~\citep{Ni&Chandler:94IEEPC}.

The theory of continuum percolation has been used by Dousse \etal in
analyzing the connectivity of a homogeneous ad hoc network under the
worst case mutual
interference~\citep{Dousse&Etal:05ITN,Dousse&Etal:06JAPPP}.
In~\citep{Kong&Yeh:08Infocom,Kong&Yeh:09ITIT}, the connectivity and
the transmission delay in a homogeneous ad hoc network with
statically or dynamically on-off links are investigated from a
percolation-based perspective.

The optimal power control in heterogeneous networks has been studied
in~\citep{Ren&etal:08JSAC}, which focuses on a single pair of
secondary users in a Poisson network of primary users. The impacts
of secondary users' transmission power on the occurrence of spectrum
opportunities and the reliability of opportunity detection are
analytically characterized.

\subsection{Organization and Notations}
\label{subsec:Organization}

The rest of this paper is organized as follows. Sec.~\ref{sec:NM}
presents the Poisson model of the heterogeneous network. In
particular, the conditions for the existence a communication link in
the secondary network is specified based on a rigorous definition of
spectrum opportunity. In Sec.~\ref{sec:AC_CR}, we introduce the
concept of connectivity region and establish its three basic
properties. The two critical densities are analyzed, followed by a
necessary and a sufficient condition for connectivity. In
Sec.~\ref{sec:PvsO}, we demonstrate the tradeoff between proximity
and opportunity by studying the impacts of the secondary users'
transmission power on the connectivity region and on the conditional
degree of a secondary user. The optimal transmission power of the
secondary users is obtained under the performance measure of the
secondary network's tolerance to the primary traffic load.
Sec~\ref{sec:proofs} contains the detailed proofs of the main
results, and Sec.~\ref{sec:conclusion} concludes the paper.

Throughout the paper, we use capital letters for parameters of the
primary users and lowercase letters for the secondary users.


\section{Network Model}
\label{sec:NM}

We consider a Poisson distributed secondary network overlaid with a
Poisson distributed primary network in an infinite two-dimensional
Euclidean space. The models of the primary and secondary networks
are specified in the following two subsections.

\subsection{The Primary Network}
\label{subsec:PN}

The primary transmitters are distributed according to a
two-dimensional Poisson point process with density $\lambda_{PT}$.
To each primary transmitter, its receiver is uniformly distributed
within its transmission range $R_p$. Here we have assumed that all
primary transmitters use the same transmission power and the
transmitted signals undergo an isotropic path loss. Based on the
displacement theorem~\citep[Chapter~5]{Kingman:Poisson}, it is easy
to see that the primary receivers form a two-dimensional Poisson
point process with density $\lambda_{PT}$. Note that the two Poisson
processes formed by the primary transmitters and receivers are
correlated.

\subsection{The Secondary Network}
\label{subsec:SN}

The secondary users are distributed according to a two-dimensional
Poisson point process with density $\lambda_S$, independent of the
Poisson processes of the primary transmitters and receivers. The
transmission range of the secondary users is denoted by $r_p$.

\subsubsection{Communication Links}
\label{subsubsec:CL}

In contrast to the case in a homogeneous network, the existence of a
communication link between two secondary users depends on not only
the distance between them but also the availability of the
communication channel (\ie the presence of a spectrum opportunity).
The latter is determined by the transmitting and receiving
activities in the primary network as described below.

As illustrated in Fig.~\ref{fig:DefSO}, there exists an opportunity
from $A$, the secondary transmitter, to $B$, the secondary receiver,
if the transmission from $A$ does not interfere with nearby {\it
primary receivers} in the solid circle, and the reception at $B$ is
not affected by nearby {\it primary transmitters} in the dashed
circle~\citep{Zhao:07ICASSP}. Referred to as the interference range
of the secondary users, the radius $r_I$ of the solid circle at $A$
depends on the transmission power of $A$ and the interference
tolerance of the primary receivers, whereas the radius $R_I$ of the
dashed circle (the interference range of the primary users) depends
on the transmission power of the primary users and the interference
tolerance of $B$.

\begin{figure}[htbp]
\centerline{
\begin{psfrags}
\psfrag{A}[c]{ $A$} \psfrag{B}[c]{ $B~$} \psfrag{i}[c]{
~~~~~~~~~~~~~~~~~Interference} \psfrag{ra}[r]{$r_I$}
\psfrag{rb}[c]{$~~R_I$} \psfrag{tx}[l]{ Primary Tx} \psfrag{rx}[l]{
Primary Rx} \scalefig{0.5}\epsfbox{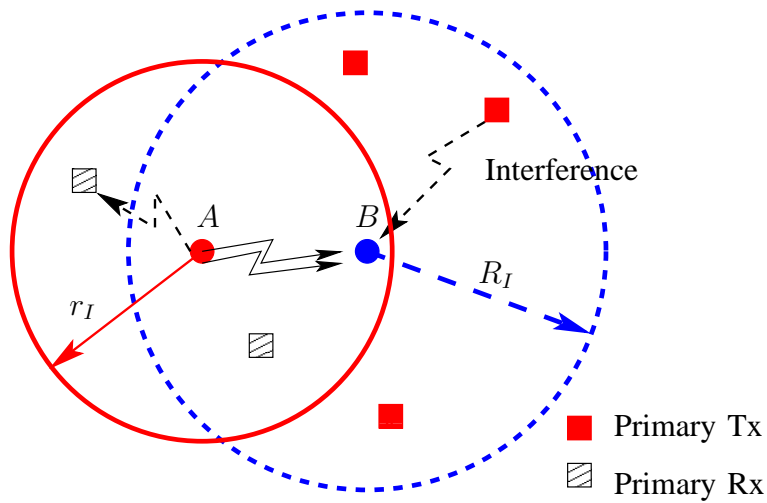}
\end{psfrags}
} \caption{Definition of spectrum opportunity.} \label{fig:DefSO}
\end{figure}

It is clear from the above discussion that spectrum opportunities
depend on both transmitting and receiving activities of the primary
users. Furthermore, spectrum opportunities are \emph{asymmetric}.
Specifically, a channel that is an opportunity when $A$ is the
transmitter and $B$ the receiver may not be an opportunity when $B$
is the transmitter and $A$ the receiver. In other words, there exist
unidirectional communication links in the secondary network. Since
unidirectional links are difficult to utilize in wireless
networks~\citep{Ramasubramanian:02INFOCOM}, we only consider
bidirectional links in the secondary network when we define
connectivity. As a consequence, when we determine whether there
exists a communication link between two secondary users, we need to
check the existence of spectrum opportunities in both directions.

To summarize, under the disk signal propagation and interference
model, there is a (bidirectional) link between $A$ and $B$ if and
only if (C1) the distance between $A$ and $B$ is at most $r_p$; (C2)
there exists a bidirectional spectrum opportunity between $A$ and
$B$, $\ie$ there are no primary transmitters within distance $R_I$
of either $A$ or $B$ and no primary receivers within distance $r_I$
of either $A$ or $B$.

\subsubsection{Connectivity}
\label{subsubsec:connectivity}

We interpret the connectivity of the secondary network in the
percolation sense: the secondary network is connected if there
exists an infinite connected component a.s.

Based on the above conditions (C1, C2) for the existence of a
communication link, we can obtain an undirected random graph
$\mathcal{G}(\lambda_S,\lambda_{PT})$ corresponding to the secondary
network, which is determined by three Poisson point processes: the
secondary users with density $\lambda_S$, the primary transmitters
with density $\lambda_{PT}$, and the primary receivers with density
$\lambda_{PT}$ (correlated to the process of the primary
transmitters)\footnote{The two Poisson point processes of the
primary transmitters and receivers are essentially a snap shot of
the realizations of the primary transmitters and receivers. In
different time slots, different sets of primary users become active
transmitters/receivers. Thus, even if a secondary user is isolated
at one time due to the absence of spectrum opportunities, it may
experience an opportunity at a different time and be connected to
other secondary users.}. See Fig.~\ref{fig:PN} for an illustration
of $\mathcal{G}(\lambda_S,\lambda_{PT})$.

\begin{figure}[htbp]
\centerline{
\begin{psfrags}
\psfrag{P_Tx}[l]{\small Primary Tx} \psfrag{P_Rx}[l]{\small Primary
Rx} \psfrag{S}[l]{\small Secondary User} \psfrag{rI}[c]{\small
$r_I$} \psfrag{RI}[c]{\small $R_I$}
\scalefig{0.7}\epsfbox{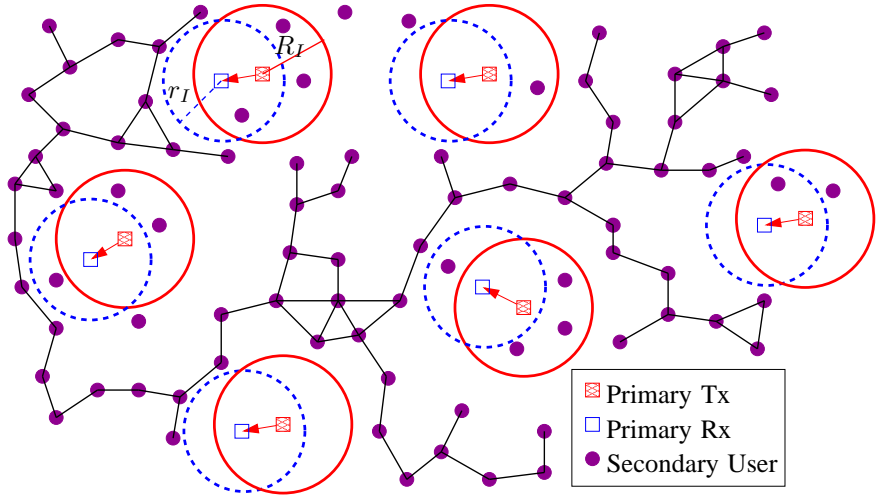}
\end{psfrags}} \caption{A
realization of the heterogeneous network. The random graph
$\mathcal{G}(\lambda_S,\lambda_{PT})$ consists of all the secondary
nodes and all the bidirectional links denoted by solid lines. The
solid circles with radii $R_I$ denote the interference regions of
the primary transmitters within which secondary users can not
successfully receive, and the dashed circles with radii $r_I$ denote
the required protection regions for the primary receivers within
which the secondary users should refrain from transmitting.}
\label{fig:PN}
\end{figure}

The question we aim to answer in this paper is the connectivity of
the secondary network, \ie the percolation in
$\mathcal{G}(\lambda_S,\lambda_{PT})$.


\section{Analytical Characterizations of the Connectivity Region}
\label{sec:AC_CR}

Given the transmission power and the interference tolerance of both
the primary and the secondary users ($\ie$ $R_p$, $R_I$, $r_p$, and
$r_I$ are fixed), the connectivity of the secondary network is
determined by the density $\lambda_S$ of the secondary users and the
density $\lambda_{PT}$ of the primary transmitters. We thus
introduce the concept of connectivity region $\mathcal{C}$ of a
secondary network, which is defined as the set of density pairs
$(\lambda_S,~\lambda_{PT})$ under which the secondary
network~$\mathcal{G}(\lambda_S,\lambda_{PT})$ is connected (see
Fig.~\ref{fig:Con_Region}).
\begin{eqnarray*}
\mathcal{C}\defeq\{(\lambda_S,~\lambda_{PT}):~\mathcal{G}(\lambda_S,\lambda_{PT})
\mbox{ is connected.}\}.
\end{eqnarray*}

\subsection{Basic Properties of the Connectivity Region}
\label{subsec:BP_CR}

We establish in Theorem~\ref{thm:BP_CR} below three basic properties
of the connectivity region.

\begin{theorem} {\it Basic Properties of the Connectivity Region.} \label{thm:BP_CR}
\begin{itemize}
\item[T1.1] The connectivity region $\Cc$ is contiguous, that is, for any
two points $(\lambda_{S1},\lambda_{PT1})$,
$(\lambda_{S2},\lambda_{PT2})\in \Cc$, there exists a continuous
path in $\Cc$ connecting the two points.
\item[T1.2] The lower boundary of the connectivity region $\Cc$ is the
$\lambda_S$-axis. Let $\lambda^*_{PT}(\lambda_S)$ denote the upper
boundary of the connectivity region $\Cc$, \ie
\[
\lambda^*_{PT}(\lambda_S)\defeq
\sup\{\lambda_{PT}:~\mathcal{G}(\lambda_S,\lambda_{PT}) \mbox{ is
connected.}\},
\]
then we have that $\lambda^*_{PT}(\lambda_S)$ is monotonically
increasing with $\lambda_S$.
\item[T1.3] There exists either zero or one infinite connected component in
$\mathcal{G}(\lambda_S,\lambda_{PT})$ a.s.
\end{itemize}
\end{theorem}

\begin{proof}
The proofs of T1.1 and T1.2 are based on the coupling argument, a
technique frequently used in continuum
percolation~\citep[Section~2.2]{Meester&Roy:Con_Percolation}. The
proof of T1.3 is based on the ergodicity of the random model driven
by the three Poisson point processes of the primary transmitters,
the primary receivers, and the secondary users (the concept of
ergodicity of a random model is reviewed in
Sec.~\ref{subsubsec:ergodic_theory}). The details of the proofs are
given in Sec.~\ref{subsec:proof_thm1}.
\end{proof}

T1.1 and T1.2 specify the basic structure of the connectivity
region, as illustrated in Fig.~\ref{fig:Con_Region}. T1.3 implies
the occurrence of a phase transition phenomenon, that is, there
exists either a unique infinite connected component a.s. or no
infinite connected component a.s. This uniqueness of the infinite
connected component establishes the strong connectivity of the
secondary network: once it is connected, it is strongly connected.
It excludes the undesirable possibility of having more than one
(maybe infinite) infinite connected component in the secondary
network. We point out that such a property is not always present in
wireless networks. Two examples where more than one infinite
connected component exists in a homogeneous ad hoc network can be
found in~\citep{Meester&Roy:94AAP}.

\subsection{Critical Densities}
\label{subsec:CD}

In this subsection, we study the critical density $\lambda^*_S$ of
the secondary users and the critical density
$\overline{\lambda^*_{PT}}$ of the primary transmitters. Recall that
\begin{eqnarray*}
\lambda^*_S &\defeq& \inf\{\lambda_S: \exists \lambda_{PT}>0 \mbox{
s.t. $\mathcal{G}(\lambda_S,\lambda_{PT})$ is connected}\}, \\
\overline{\lambda^*_{PT}} &\defeq& \sup\{\lambda_{PT}: \exists
\lambda_S<\infty \mbox{ s.t. $\mathcal{G}(\lambda_S,\lambda_{PT})$
is connected}\}.
\end{eqnarray*}
We have the following theorem.

\begin{theorem} {\it Critical Densities.}
\label{thm:critical_density} \\
Given $R_p$, $R_I$, $r_p$, and $r_I$, we have
\begin{itemize}
\item[T2.1] $\lambda^*_S=\lambda_c(r_p)$, where $\lambda_c(r_p)$ is the critical
density for a homogeneous ad hoc network with transmission range
$r_p$ (\ie in the absence of the primary network).
\item[T2.2] $\overline{\lambda^*_{PT}}\leq \frac{\lambda_c(1)}{4\max \{R_I^2,
r_I^2\}-r_p^2}$, where the constant $\lambda_c(1)$ is the critical
density for a homogeneous ad hoc network with a unit transmission
range.
\end{itemize}
\end{theorem}

\begin{proof}
The basic idea of the proof of T2.1 is to approximate the secondary
network $\mathcal{G}(\lambda_S,\lambda_{PT})$ by a discrete
edge-percolation model on the grid. This discretization technique is
often used to convert a continuum percolation model to a discrete
site/edge percolation model (see, for example,~\citep[Chapter
3]{Meester&Roy:Con_Percolation},~\citep{Dousse&Etal:06JAPPP}). The
details of the proof are given in Sec.~\ref{subsubsec:proof_thm21}.

The proof of T2.2 is based on the argument that if there is an
infinite connected component in the secondary network, then an
infinite vacant component must exist in the two Poisson Boolean
models driven by the primary transmitters and the primary receivers,
respectively. The key point is to carefully choose the radii of the
two Poisson Boolean models in order to obtain a valid upper bound on
$\overline{\lambda^*_{PT}}$. The details of the proof can be found
in Sec.~\ref{subsubsec:proof_thm22}.
\end{proof}

\begin{figure}[htbp]
\centerline{\scalefig{0.85} \epsfbox{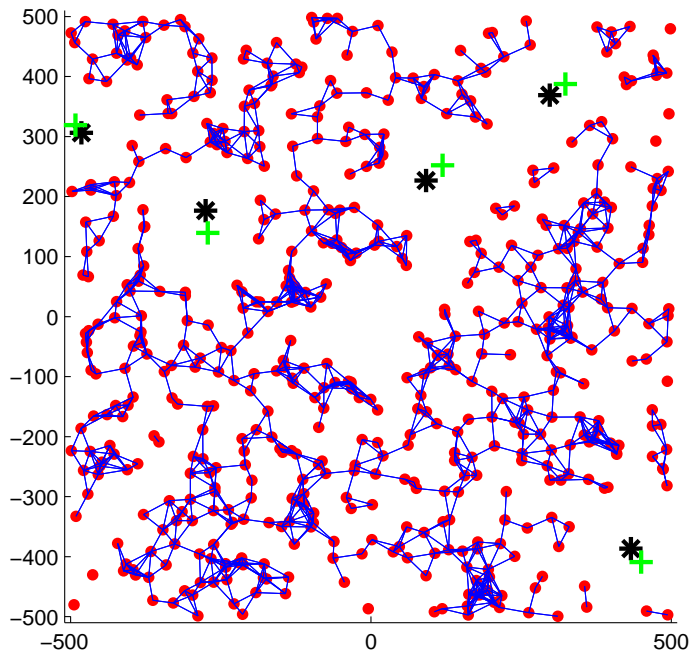}}
\caption{A realization of the Poisson heterogeneous network when the
percolation occurs (black stars denote primary transmitters, green
plus signs denote primary receivers, red dots denote secondary
users, and blue segments denote the bidirectional links between
secondary users). We have removed secondary users who do not see
opportunities for clarity. The simulation parameters are given by
$\lambda_{PT}=10$km$^{-2}$, $R_p=50$m, $R_I=80$m,
$\lambda_S=650$km$^{-2}$, $r_p=50$m, $r_I=80$m, and the critical
density in this case is $\lambda_c(50)\approx 576$km$^{-2}$.}
\label{fig:Connected_Network}
\end{figure}

Fig.~\ref{fig:Connected_Network} shows one realization of the
Poisson heterogeneous network when $\lambda_S$ is slightly larger
than $\lambda_c(r_p)$ and $\lambda_{PT}$ is small. At least one
left-to-right (L-R) crossing and at least one top-to-bottom (T-B)
crossing can be found in the square network. It is thus expected
that these L-R and T-B crossings in finite square regions can form
an infinite connected component in the whole network on
$\mathbb{R}^2$. If we slightly increase $\lambda_{PT}$, then we
observe from Fig.~\ref{fig:Disconnected_Network} that the reduction
in spectrum opportunities eliminates considerable communication
links in the secondary network, creating several disjoint small
components.

\begin{figure}[htbp]
\centerline{\scalefig{0.85}
\epsfbox{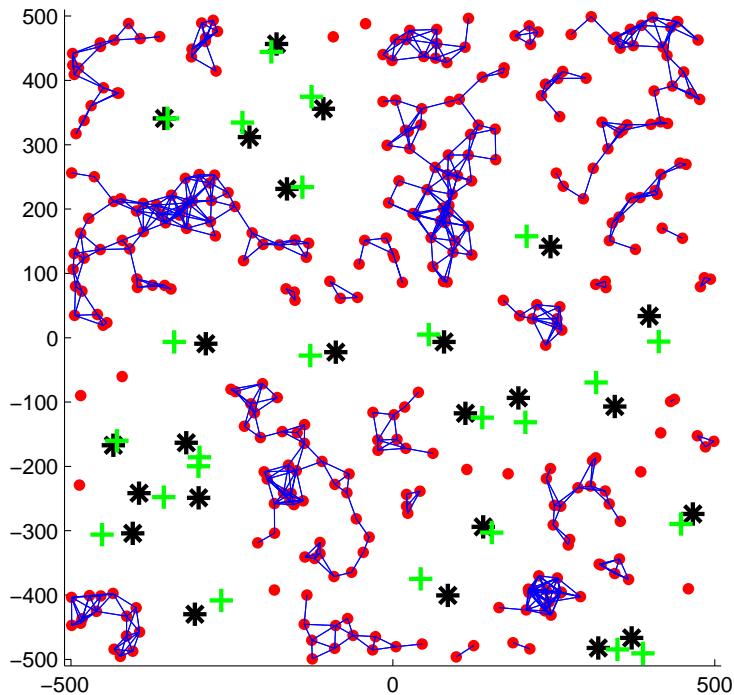}} \caption{A realization
of the Poisson heterogeneous network when the percolation does not
occur (black stars denote primary transmitters, green plus signs
denote primary receivers, red dots denote secondary users, and blue
segments denote the bidirectional links between secondary users). We
have removed secondary users who do not see opportunities for
clarity. The simulation parameters are given by
$\lambda_{PT}=20$km$^{-2}$, $R_p=50$m, $R_I=80$m,
$\lambda_S=650$km$^{-2}$, $r_p=50$m, $r_I=80$m, and the critical
density in this case is $\lambda_c(50)\approx 576$km$^{-2}$.}
\label{fig:Disconnected_Network}
\end{figure}

Fig.~\ref{fig:Con_Region_Simu2} shows a simulation example of the
connectivity region, where the upper bound on the critical density
$\overline{\lambda^*_{PT}}$ of the primary transmitters given in
T2.2 appears to be achievable.

\begin{figure}[htbp]
\centerline{
\scalefig{0.7}\epsfbox{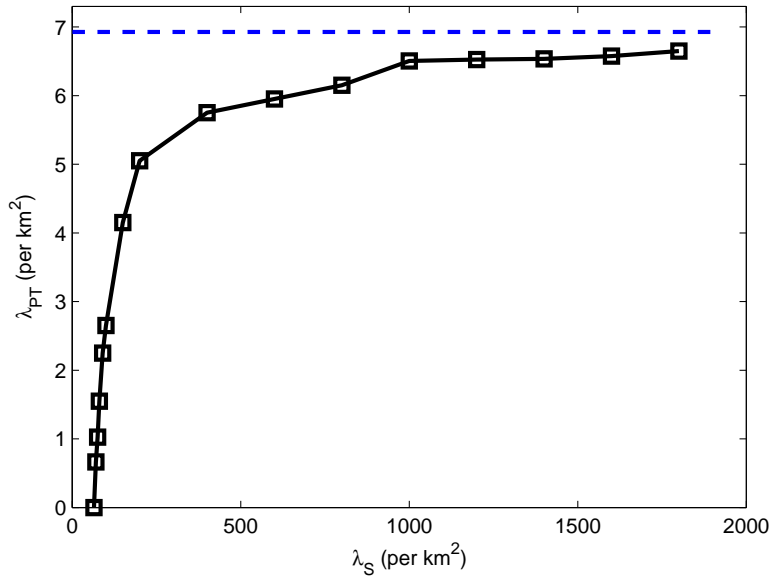} }
\caption{Simulated connectivity regions when $r_p = 150$m, $r_I =
240$m, $R_p =100$m, and $R_I=120$m. The blue dashed line is the
upper bound $\frac{\lambda_c(1)}{4\max \{R_I^2, r_I^2\}-r_p^2}$ on
the critical density $\overline{\lambda^*_{PT}}$ of primary
transmitters given in T2.2. The area of the simulated heterogeneous
network is $2000$m$\times 2000$m. For a fixed density $\lambda_S$ of
the secondary users, the upper boundary $\lambda^*_{PT}(\lambda_S)$
is equal to the minimum density of the primary transmitters such
that over all the $1000$ realizations, the percentage of the ones in
which there exists at least one L-R crossing is below $50$\%. The
intuitive reason for choosing the existence of an L-R crossing as
the criterion for connectivity is illustrated in
Fig.~\ref{fig:Connected_Network}-\ref{fig:Disconnected_Network}.}
\label{fig:Con_Region_Simu2}
\end{figure}

\subsection{A Necessary Condition for Connectivity}
\label{subsec:NC_connectivity}

In this subsection, we establish a necessary condition for
connectivity which is given in terms of the average conditional
degree of a secondary user. This condition agrees with our
intuition: the secondary network cannot be connected if the degree
of every secondary user is small.

Let $\mathbb{I} (A,d,\textrm{rx/tx})$ denote the event that there
exists primary receivers/transmitters within distance $d$ of a
secondary user $A$. Let $\overline{\mathbb{I}(A,d,\textrm{rx/tx})}$
denote the complement of $\mathbb{I}(A,d,\textrm{rx/tx})$. Since a
secondary user is isolated if it does not see a spectrum
opportunity, we focus on secondary users who experience spectrum
opportunities and define the conditional average degree $\mu$ of
such a secondary user $A$ as
\begin{eqnarray} \label{eqn:def_cond_avg_deg}
\mu =
\mathbb{E}[deg(A)|~\overline{\mathbb{I}(A,r_I,\textrm{rx})}\cap
\overline{\mathbb{I}(A,R_I,\textrm{tx})}],
\end{eqnarray}
where $deg(A)$ denotes the degree of $A$, $r_I$ the interference
range of the secondary users, and $R_I$ the interference range of
the primary users. Notice that the degree of $A$ is the number of
secondary users within the transmission range of $A$ \emph{and}
experiencing opportunities. We arrive at the following necessary
condition for connectivity.

\begin{theorem} \label{thm:necessary_cond}
A necessary condition for the connectivity of
$\mathcal{G}(\lambda_S,\lambda_{PT})$ is $\mu>1$, where $\mu$ is the
conditional average degree of a secondary user defined
in~(\ref{eqn:def_cond_avg_deg}).
\end{theorem}

\begin{proof}
The basic idea is to construct a branching process, where the
conditional average degree $\mu$ is the average number of offspring.
This branching process provides an upper bound on the number of
secondary users in a connected component. If $\mu \leq 1$, then the
branching process is finite a.s. It thus follows that there is no
infinite connected component a.s. in
$\mathcal{G}(\lambda_S,\lambda_{PT})$. Details can be found in
Sec.~\ref{subsec:proof_thm3}.
\end{proof}

To apply the necessary condition given in
Theorem~\ref{thm:necessary_cond}, the conditional average degree
$\mu$ of a secondary user $A$ needs to be evaluated based on the
network parameters. Let $B$ be a secondary user randomly and
uniformly distributed within the transmission range $r_p$ of $A$.
Let $g(\lambda_{PT},r_p,r_I,R_p,R_I)$ denote the probability of a
bidirectional opportunity between $A$ and $B$ conditioned on the
event that $A$ sees an opportunity. Based on the statistical
equivalence and independence of different points in a Poisson point
process, the conditional average degree $\mu$ of a secondary user
$A$ is given by this conditional probability $g(\cdot)$ of a
bidirectional opportunity between $A$ and a randomly chosen neighbor
multiplied by the average number of neighbors of $A$, \ie
\begin{eqnarray} \label{eqn:cond_avg_deg_sim}
\mu = \left(\lambda_S \pi r_p^2\right)\cdot
g(\lambda_{PT},r_p,r_I,R_p,R_I).
\end{eqnarray}
The detailed derivation for $(\ref{eqn:cond_avg_deg_sim})$ and the
expression for $g(\cdot)$ are given in Appendix~A. It is also shown
in Appendix~A that $g(\cdot)$ is a strictly decreasing function of
$\lambda_{PT}$. Thus $g^{-1}(\cdot)$, the inverse of $g(\cdot)$ with
respect to $\lambda_{PT}$, is well-defined.

Combining (\ref{eqn:cond_avg_deg_sim}) with
Theorem~\ref{thm:necessary_cond}, we obtain an outer bound on the
connectivity region. Specifically, let $\mu
(\lambda_S,\lambda_{PT})$ denote the conditional average degree of a
secondary user in $\mathcal{G}(\lambda_S,\lambda_{PT})$. Then those
density pairs $(\lambda_S,\lambda_{PT})$ satisfying $\mu
(\lambda_S,\lambda_{PT})\leq 1$ are outside the connectivity region.

\begin{corollary} \label{cor:outer_bound}
Given $R_p$, $R_I$, $r_p$, and $r_I$, an outer bound on the
connectivity region $\Cc$ is given by
\begin{eqnarray*}
\lambda_{PT}=g^{-1}\left(\frac{1}{\lambda_S \pi r_p^2}\right),
\end{eqnarray*}
where $g^{-1}(\cdot)$ is the inverse of the conditional probability
$g(\cdot)$ with respect to $\lambda_{PT}$.
\end{corollary}

\subsection{A Sufficient Condition for Connectivity}
\label{subsec:SC_connectivity}

In this subsection, we establish a sufficient condition for
connectivity, which provides an inner bound on the connectivity
region and a criterion for checking whether a secondary network is
connected.

\begin{figure}[htbp]
\centerline{
\begin{psfrags}
\psfrag{d}[c]{$d$} \psfrag{O}[c]{$O$} \psfrag{O1}[c]{$O_1$}
\psfrag{O2}[c]{$O_2$} \psfrag{O3}[c]{$O_3$} \psfrag{O4}[c]{$O_4$}
\psfrag{O5}[c]{$O_5$} \psfrag{O6}[c]{$O_6$}
\psfrag{O7}[c]{$O_7$}\psfrag{O8}[c]{$O_8$}
\scalefig{0.65}\epsfbox{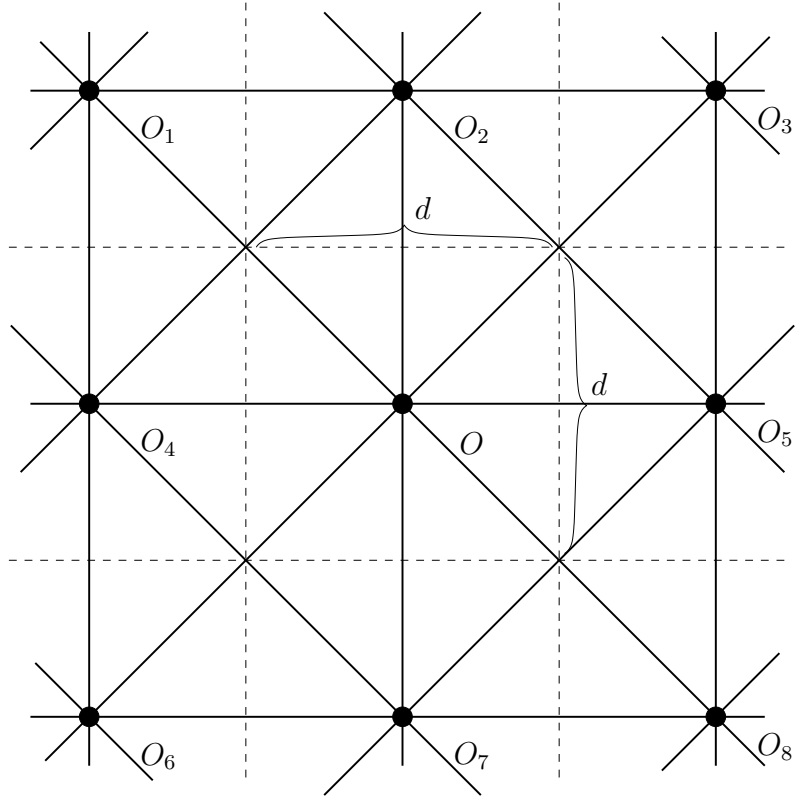}
\end{psfrags}
} \caption{An illustration of the dependent site-percolation model
$\mathcal{L}$ with side length $d$ (solid dots denote sites, solid
lines denote edges connecting every two sites, and dashed lines
denote the squared partition).} \label{fig:D_Site_Per}
\end{figure}

The sufficient condition for connectivity is established by using
the discretization technique. The continuum percolation model is
mapped onto a dependent site-percolation model $\mathcal{L}$ in the
following way. As illustrated in Fig.~\ref{fig:D_Site_Per}, we
partition $\mathbb{R}^2$ into (dashed) squares with side length $d$
and locate a site at the center of each square. Sites whose
associated dashed squares share at least one common point are
considered connected (as illustrated by solid lines in
Fig.~\ref{fig:D_Site_Per}). Thus each site is connected to eight
neighbors\footnote{For the commonly used square site-percolation
model, each site has four neighbors. The site-percolation model
constructed here can provide a better inner bound.} (see the eight
neighbors $O_1$,...,$O_8$ of site $O$ in Fig.~\ref{fig:D_Site_Per}).
Let $B_O$ be the associated dashed square of $O$, then $O$ is
occupied if there exists in $B_O$ at least one secondary user who
sees an opportunity.

Since the largest distance between two points in two neighboring
dashed squares is $2\sqrt{2}d$, it follows that if we set
$d=\frac{r_p}{2\sqrt{2}}$, then for every pair of secondary users in
two neighboring dashed squares, they are within the transmission
range $r_p$ of each other. Based on the definitions of occupied site
in $\Lc$ and communication link in the secondary network, we
conclude that the existence of an infinite occupied component (a
connected component consisting of only occupied sites) in
$\mathcal{L}$ implies the existence of an infinite connected
component in the secondary network.

Due to the fact that spectrum opportunities are spatially dependent,
the state of one site is correlated with the states of its adjacent
sites. Thus, the above site-percolation model $\mathcal{L}$ is a
dependent model. Define the dependence range $k$ as the minimum
distance such that the state of any two sites at distance $d>k$ are
independent, where the distance between two sites is the minimum
number of neighboring sites that must be traversed from one site to
the other. Then the dependence range of $\mathcal{L}$ is given by
\begin{eqnarray} \label{eqn:dependency_range}
k=\left\lceil \frac{8 \max
\left\{R_I+\frac{r_p}{4},r_I+\frac{r_p}{4}\right\}}{r_p}\right\rceil-1.
\end{eqnarray}

Let $p_c$ denote the upper critical probability of $\Lc$ which is
defined as the minimum occupied probability $p^*$ such that if the
occupied probability $p>p^*$, an infinite occupied component
containing the origin exists in $\mathcal{L}$ with a positive
probability (wpp.). Since the dependence range $k$ of $\mathcal{L}$
is finite, it follows from Theorem
2.3.1~\citep{Franceschetti&Meester:Random_Network_Comm} that
$p_c<1$. Now we present the sufficient condition for connectivity in
the following theorem.

\begin{theorem}
Let $p_c$ denote the upper critical probability of the dependent
site-percolation model $\mathcal{L}$ specified above. Define
\begin{eqnarray} \label{eqn:I_def}
I(r,R_p,r_I)=2\int_0^r t \frac{S_I (t,R_p,r_I)}{\pi
R_p^2}\mathrm{d}t,
\end{eqnarray}
where $S_I(t,R_p,r_I)$ is the common area of two circles with radii
$R_p$ and $r_I$ and centered $t$ apart. Then the secondary network
is connected if
\begin{eqnarray*}
\left[1-\exp \left(-\frac{\lambda_S r_p^2}{8}\right)\right]\exp
\left\{-\lambda_{PT} \pi \left[R_I^2+r_I^2
-I\left(R_I,R_p,r_I\right)\right]\right\}>p_c.
\end{eqnarray*}
\end{theorem}

\begin{proof}
The proof is based on the ergodicity of the heterogeneous network
model and its relation with the constructed dependent
site-percolation model $\Lc$. Details can be found in
Sec.~\ref{subsec:proof_thm4}.
\end{proof}

By applying a general upper bound on the upper critical probability
$p_c$ for a site-percolation model with finite dependence
range~\citep[Theorem
2.3.1]{Franceschetti&Meester:Random_Network_Comm}, we arrive at the
following corollary.
\begin{corollary} \label{cor:inner_bound}
A sufficient condition for the connectivity of
$\mathcal{G}(\lambda_S,\lambda_{PT})$ is
\begin{eqnarray*}
\lambda_{PT}<\frac{1}{\pi \left[R_I^2+r_I^2
-I(R_I,R_p,r_I)\right]}\ln \frac{1-\exp \left(-\frac{\lambda_S
r_p^2}{8}\right)}{1-\left(\frac{1}{3}\right)^{(2k+1)^2}},
\end{eqnarray*}
where $I(R_I,R_p,r_I)$ is defined in~(\ref{eqn:I_def}) and $k$ is
the dependence range of the site-percolation model defined
in~(\ref{eqn:dependency_range}).
\end{corollary}


\section{Impact of Transmission Power: Proximity vs. Opportunity}
\label{sec:PvsO}

In this section, we study the impact of the secondary users'
transmission power on the connectivity and the conditional average
degree of the secondary network. As has been illustrated in
Fig.~\ref{fig:Con_Region_Simu1}, there exists a tradeoff between
proximity and opportunity in designing the secondary users'
transmission power for connectivity. Specifically, increasing the
transmission power of the secondary users leads to a smaller
critical density $\lambda^*_S$ of the secondary users, but at the
same time, a lower tolerance to the primary traffic load manifested
by a smaller critical density $\overline{\lambda^*_{PT}}$ of the
primary transmitters.

\subsection{Impact on the Conditional Average Degree}
\label{subsec:Imp_Cond_Avg_Deg}

As discussed in Sec.~\ref{subsec:NC_connectivity}, the expression
for the conditional average degree $\mu$ can be decomposed into the
product of two terms: $\lambda_S \pi r_p^2$ and
$g(\lambda_{PT},r_p,r_I,R_p,R_I)$. The first term is the average
number of neighbors of a secondary user, which increases with the
transmission power $p_{tx}$ of the secondary users (\ie enhanced
proximity). The other term $g(\lambda_{PT},r_p,r_I,R_p,R_I)$ is the
conditional probability of a bidirectional opportunity, which
decreases with $p_{tx}$ due to reduced spectrum opportunities. This
tension between proximity and opportunity is illustrated in
Fig.~\ref{fig:Con_Avg_Deg}, where we observe that the impact of
$p_{tx}$ on proximity dominates when $p_{tx}$ is small ($\mu$
increases with $p_{tx}$) while its impact on the occurrence of
opportunities dominates when $p_{tx}$ is large ($\mu$ decreases with
$p_{tx}$).

\begin{corollary} \label{cor:cond_avg_deg}
Let $p_{tx}$ be the transmission power of secondary users and $\mu$
the conditional average degree defined in
(\ref{eqn:def_cond_avg_deg}), then under the disk signal propagation
and interference model we have\footnote{Here we use the Big O
notation: $f(x)=O(g(x))$ as $x\rightarrow \infty$ if and only if
$\exists~M>0$, $x_0>0$ such that $|f(x)|\leq M|g(x)|$ for all
$x>x_0$.}
\begin{eqnarray*}
\mu = O\left((p_{tx})^{-2/\alpha}\right)\textrm{ as }
p_{tx}\rightarrow \infty,
\end{eqnarray*}
where $\alpha$ is the path-loss exponent.
\end{corollary}

\begin{proof}
We show this corollary by deriving an upper bound on the conditional
average degree $\mu$. Details can be found in Appendix~B.
\end{proof}

For a homogeneous network, the average degree of a user is $\lambda
\pi r_p^2$, which increases with $p_{tx}$ at rate
$\left(p_{tx}\right)^{2/\alpha}$. In sharp contrast, this corollary
tells us that for a heterogeneous network, when $p_{tx}$ is large
enough, the conditional average degree $\mu$ of a secondary user
actually \emph{decreases} with $p_{tx}$ at least as fast as
$\left(p_{tx}\right)^{-2/\alpha}$.

\begin{figure}[htbp]
\centerline{ \scalefig{0.7}\epsfbox{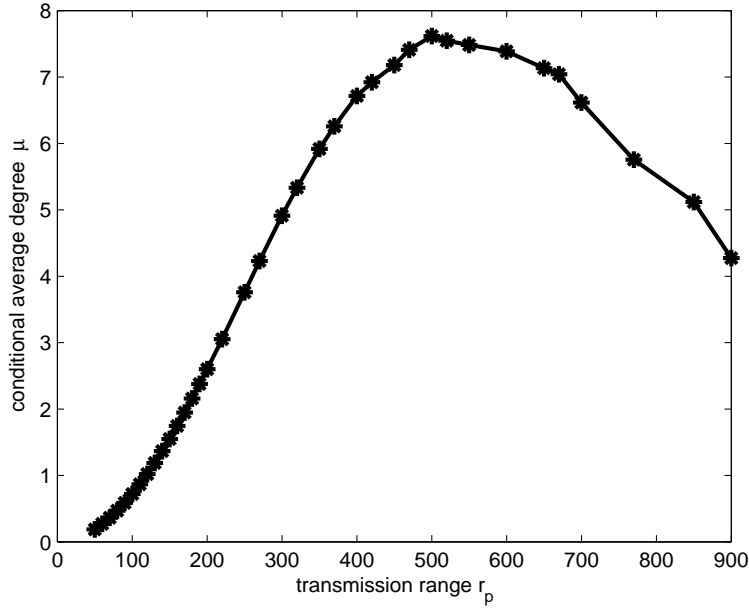}
} \caption{Conditional average degree $\mu$ of secondary users vs
transmission range $r_p$ of secondary users ($r_p\propto
\left(p_{tx}\right)^{\frac{1}{\alpha}}$, where $p_{tx}$ is the
transmission power of secondary users and $\alpha$ is the path-loss
exponent, and simulation parameters are given by $\lambda_{PT} =
2.5$km$^{-2}$, $R_p = 200$m, $R_I = 250$m, $\lambda_S =
25$km$^{-2}$, $r_I = r_p/0.8$).} \label{fig:Con_Avg_Deg}
\end{figure}

\subsection{Impact on the Connectivity Region}
\label{subsec:Imp_CD}

From the scaling relation of the critical density~\citep[Proposition
2.11]{Meester&Roy:Con_Percolation}, we know that in a homogeneous
two-dimensional network,
\begin{eqnarray*}
\lambda_c(r_p)=\lambda_c(1)\left(r_p\right)^{-2}\propto
\left(p_{tx}\right)^{-\frac{2}{\alpha}},
\end{eqnarray*}
where the constant $\lambda_c(1)$ is the critical density for a
homogeneous ad hoc network with a unit transmission range. Thus, if
each secondary user adopts a high transmission power, then
$\lambda_c(r_p)$ reduces. It follows from T2.1 that the critical
density $\lambda^*_S$ of secondary users to achieve connectivity
reduces due to the enhanced proximity.

On the other hand, from the upper bound on the critical density
$\overline{\lambda^*_{PT}}$ of the primary transmitters given in
T2.2, we have that
\begin{eqnarray*}
\overline{\lambda^*_{PT}} =
O\left((p_{tx})^{-2/\alpha}\right)\textrm{ as } p_{tx}\rightarrow
\infty,
\end{eqnarray*}
where we have assumed that $r_p=\beta r_I$ for some $\beta \in
(0,1)$ under the disk signal propagation and interference
model\footnote{Since the minimum received signal power required for
successful reception is, in general, higher than the maximum
allowable received interference power , the transmission range $r_p$
is smaller than the interference range $r_I$, \ie $\beta<1$.}. Thus,
when the transmission power $p_{tx}$ of the secondary network is
large enough, the critical density $\overline{\lambda^*_{PT}}$ of
the primary transmitters decreases with $p_{tx}$ at least as fast as
$\left(p_{tx}\right)^{-2/\alpha}$ due to reduced spectrum
opportunities.

\subsection{Optimal Design of Transmission Power}
\label{subsec:ODTP}

Due to the tension between proximity and opportunity, there does not
exist a transmission power of the secondary users that leads to the
``largest'' connectivity region (largest in the sense that its
connectivity region contains all regions achievable with any finite
transmission power $p_{tx}$ of the secondary users). Thus, the
optimal design of $p_{tx}$ depends on the operating point of the
heterogeneous network. For instance, when a sparse secondary network
is overlaid with a primary network with low traffic load, a large
$p_{tx}$ may be desirable to achieve connectivity. The opposite
holds when a dense secondary network is overlaid with a primary
network with high traffic load.

Focusing on a sufficiently dense secondary network, we address the
design of its transmission power for the maximum tolerance to the
primary traffic. Due to its tractability and achievability indicated
by simulation examples (see Fig.~\ref{fig:Con_Region_Simu2}), the
upper bound on the critical density $\overline{\lambda^*_{PT}}$ of
the primary transmitters given in T2.2 is used as the performance
measure.

\begin{theorem} \label{thm:design_tx_power}
Let $r_I$ and $R_I$ denote the interference range of the secondary
and the primary users, respectively. For a fixed $R_I$, the upper
bound on $\overline{\lambda^*_{PT}}$ given in T2.2 is maximized when
the primary and secondary networks have matching interference
ranges: $r_I=R_I$.
\end{theorem}

\begin{proof}
Since under the disk signal propagation and interference model,
$r_p=\beta r_I$ for some $\beta \in (0,1)$, the upper bound on
$\overline{\lambda^*_{PT}}$ can be written as
\begin{eqnarray*}
\overline{\lambda^*_{PT}}\leq \left\{
\begin{array}{ll}
\frac{\lambda_c (1)}{4R_I^2-\beta^2 r_I^2} & \textrm{ for $r_I\leq
R_I$,} \\
\frac{\lambda_c (1)}{(4-\beta^2) r_I^2} & \textrm{ for $r_I> R_I$.}
\end{array}
\right.
\end{eqnarray*}
Then the above theorem can be readily shown by finding the maximal
point for the two cases: $r_I\leq R_I$ and $r_I>R_I$.
\end{proof}

An example of the upper bound on $\overline{\lambda^*_{PT}}$ is
plotted as a function of $r_I$ in
Fig.~\ref{fig:Upper_Bound_lambda_P}. Notice that there is a distinct
difference in the slope on the two sides of the optimal point. As a
consequence, the operating region of $r_I<R_I$ is preferred over
that of $r_I>R_I$ when the optimal point $r_I=R_I$ cannot be
achieved. We point out that the desired operating region of
$r_I<R_I$ is the typical case of a secondary network coexisting with
a privileged primary network.

\begin{figure}[htbp]
\centerline{\scalefig{0.65}
\epsfbox{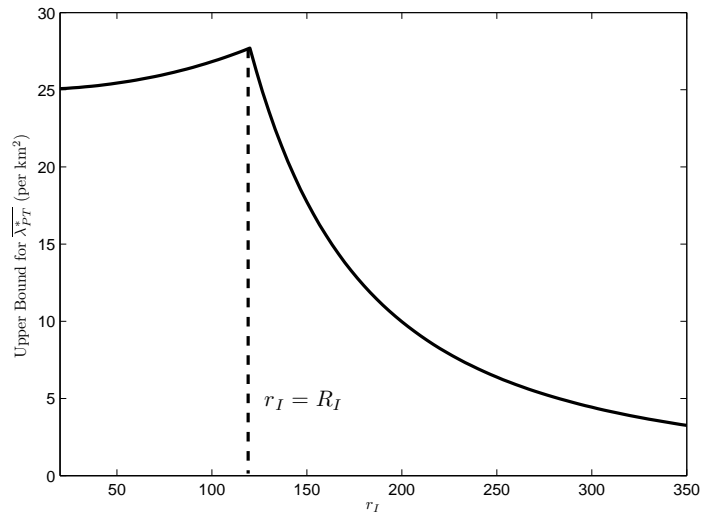}} \caption{An example of
the upper bound on $\overline{\lambda^*_{PT}}$ as a function of
$r_I$ (Parameters are given by $R_I=120$m, $r_p=0.625r_I$).}
\label{fig:Upper_Bound_lambda_P}
\end{figure}


\section{Proofs}
\label{sec:proofs}

In this section, we present proofs of the main results presented in
Sec.~\ref{sec:AC_CR}-\ref{sec:PvsO}. We start with a brief overview
of several basic results in percolation and ergodic theory that will
be used in the proofs.

\subsection{Percolation and Ergodic Theory}
\label{subsec:CP}

\subsubsection{Poisson Boolean Model}
\label{subsubsec:PBM}

Poisson Boolean model is a common model in continuum
percolation~\citep{Meester&Roy:Con_Percolation}. Often referred to
as $\mathcal{B}(X,~\rho,~\lambda)$, the model is specified by two
elements: a Poisson point process $X$ on $\mathbb{R}^d$ with density
$\lambda$ and a radius random variable $\rho$ with a given
distribution. Under this model, each point in $X$ is the center of a
circle in $\mathbb{R}^d$ with a random radius distributed according
to the distribution of $\rho$. Radii associated with different
points are independent, and they are also independent of points in
$X$. Under a Poisson Boolean model, the whole space is partitioned
into two regions: the occupied region, which is the region covered
by at least one ball, and the vacant region, which is the complement
of the occupied region. We define occupied (vacant) components as
those connected components in the occupied (vacant) region.

Assume that nodes in a homogeneous ad hoc network form a Poisson
point process with density $\lambda$ and their transmission range is
$r$. It is easy to see that the connectivity of this network can be
studied through examining the occupied connected components in the
corresponding Poisson Boolean model $\mathcal{B}(X,~r/2,~\lambda)$.

\subsubsection{Sharp Transition in Two Dimensions}
\label{subsubsec:STTD}

Phase transition is a well-known phenomenon in percolation. For the
Poisson Boolean model in two dimensions, this phenomenon appears
more remarkable in the sense that the critical density for the a.s.
existence of infinite occupied components is equal to that for the
a.s. existence of infinite vacant components. Let $\lambda_c(2\rho)$
denote the critical density for the Poisson Boolean model
$\mathcal{B}(X,~\rho,~\lambda)$, then we have that
\begin{itemize}
\item[$\Box$] when $\lambda <\lambda_c(2\rho)$, there is no infinite
occupied component a.s. and there is a unique infinite vacant
component a.s.;
\item[$\Box$] when $\lambda >\lambda_c(2\rho)$, there is a unique infinite
occupied component a.s. and there is no infinite vacant component
a.s.
\end{itemize}
The exact value of $\lambda_c$ is not known. For a deterministic
radius $\rho$, simulation results~\citep{Quintanilla&Etal:00JPA}
indicate that $\lambda_c(2\rho)\approx 0.36\rho^{-2}$, while
rigorous bounds $0.192\rho^{-2}<\lambda_c(2\rho)<0.843\rho^{-2}$ are
provided in~\citep{Meester&Roy:Con_Percolation, Kong&Yeh:07ISIT}.

\subsubsection{Crossing Probabilities}
\label{subsubsec:Crossing_Probs}

A continuous curve in the occupied region is called an occupied
path. An occupied path $\gamma$ is an occupied L-R crossing of the
rectangle $\{0\leq x\leq l_1\}\times \{0\leq y\leq l_2\}$ if
$\gamma$ intersects with both the left and the right boundaries of
the rectangle, \ie $\gamma \cap (\{x=0\}\times \{0\leq y\leq
l_2\})\neq \phi$, $\gamma \cap (\{x=l_1\}\times \{0\leq y\leq
l_2\})\neq \phi$, and the segment between the two intersecting
points is fully contained in the rectangle (see
Fig.~\ref{fig:Crossing}(a)). Similarly, we define an occupied T-B
crossing by requiring that $\gamma$ intersects with the top and
bottom boundaries of the rectangle (see Fig.~\ref{fig:Crossing}(b)).
Let
\begin{eqnarray*}
\sigma((l_1,l_2),~\lambda,~\textrm{L-R})&=&\textrm{Pr}\{\exists\textrm{
an occupied L-R
crossing of $[0,l_1]\times [0,l_2]$}\}, \\
\sigma((l_1,l_2),~\lambda,~\textrm{T-B})&=&\textrm{Pr}\{\exists\textrm{
an occupied T-B crossing of $[0,l_1]\times [0,l_2]$}\},
\end{eqnarray*}
denote the two crossing probabilities in the rectangle
$[0,l_1]\times [0,l_2]$. Then for a Poisson Boolean model
$\mathcal{B}(X,~\rho,~\lambda)$ in two dimensions with a.s. bounded
$\rho$, we have~\citep[Corollary 4.1]{Meester&Roy:Con_Percolation}
that for any $k\geq 1$,
\begin{equation} \label{eqn:crossing_prob}
\underset{n\rightarrow \infty}{\lim}\sigma
((kn,n),~\lambda,~\textrm{L-R})=\left\{
\begin{array}{ll}
1, & \textrm{ if }\lambda >\lambda_c(2\rho); \\
0, & \textrm{ if }\lambda <\lambda_c(2\rho).
\end{array}
\right.
\end{equation}
Due to the symmetry of the Poisson Boolean model, similar results
hold for the T-B crossing probability
$\sigma((n,kn),~\lambda,~\textrm{T-B})$.

\begin{figure}[htbp]
\centerline{
\begin{psfrags}
\psfrag{0}[c]{$0$} \psfrag{l1}[c]{$l_1$} \psfrag{l2}[c]{$l_2$}
\psfrag{gamma}[c]{$\gamma$} \psfrag{x}[c]{$x$} \psfrag{y}[c]{$y$}
\psfrag{(a)}[c]{(a)} \psfrag{(b)}[c]{(b)}
\scalefig{0.51}\epsfbox{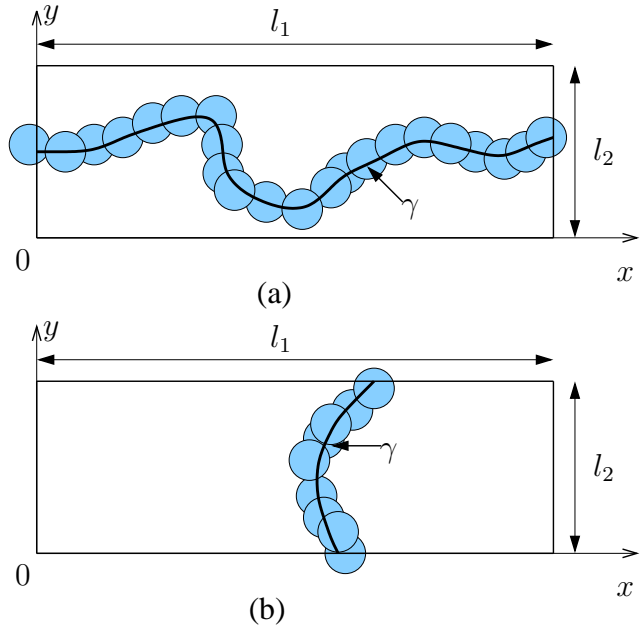}
\end{psfrags}
} \caption{An illustration of the L-R crossing (a) and the T-B
crossing (b) in a rectangle $\{0\leq x\leq l_1\}\times \{0\leq y\leq
l_2\}$.} \label{fig:Crossing}
\end{figure}

\subsubsection{Dependent Edge-Percolation Model}
\label{subsubsec:DEPM}

Let $\mathcal{L}$ be a square lattice on $\mathbb{R}^2$ with side
length $d$ (see Fig.~\ref{fig:Dual_L}). In an edge-percolation
model, every site in $\mathcal{L}$ is occupied but every edge in
$\mathcal{L}$ exists with some probability $p$. An existing edge is
often referred to as an open edge, and an edge that is not open is
called closed. When the states (open/closed) of edges are
correlated, we have a dependent edge percolation model.

\begin{figure}[htbp]
\centerline{
\begin{psfrags}
\psfrag{O}[c]{$\mathbf{O}$} \psfrag{L}[c]{$\mathbf{\mathcal{L}}$}
\psfrag{L+}[c]{$\mathbf{\mathcal{L}^+}$} \psfrag{d}[c]{$d$}
\psfrag{d2}[c]{$\frac{d}{2}$}
\scalefig{0.65}\epsfbox{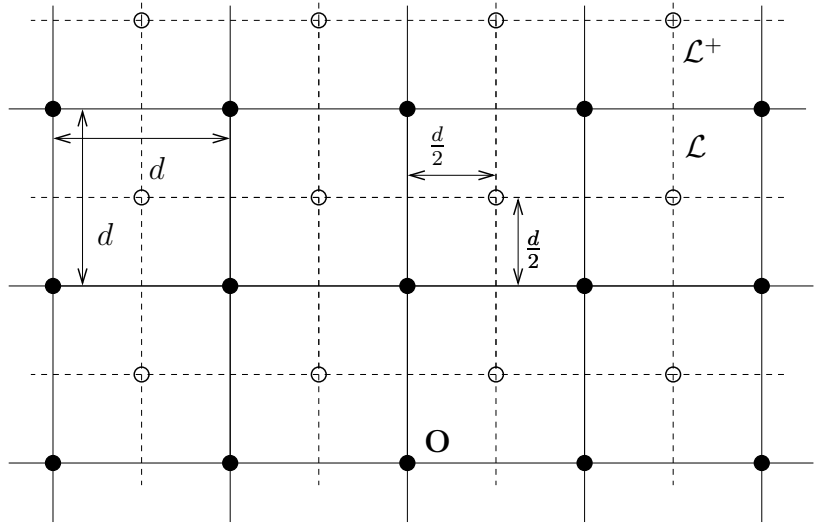}
\end{psfrags}
} \caption{Part of the lattice $\mathcal{L}$ together with its dual
$\mathcal{L}^+$ (solid dots and solid segments are sites and edges
in $\mathcal{L}$, and hollow dots and dashed segments are sites and
edges in $\mathcal{L}^+$). The dual lattice $\mathcal{L}^+$ is the
$\left(\frac{d}{2},~\frac{d}{2}\right)$-shifted version of
$\mathcal{L}$, which is used in the proof of T2.1. Since distinct
edges in $\mathcal{L}$ are crossed by distinct edges in
$\mathcal{L}^+$ and vice versa, there is a one-to-one mapping from
the edges of $\mathcal{L}$ to the edges of $\mathcal{L}^+$. In this
case, we claim an edge in $\mathcal{L}^+$ being open if and only if
its corresponding edge ($\ie$ the edge that it crosses) in
$\mathcal{L}$ is open.} \label{fig:Dual_L}
\end{figure}

Consider a special case of dependent edge-percolation model
$\mathcal{L}$ where the state of an edge $e$ is only correlated with
its six adjacent edges (edges that share a common point with $e$).
We have the following known result.

\begin{fact}~\citep[Proposition 1]{Dousse&Etal:06JAPPP} \label{fact:upper_bound_n_closed} \\
For any collection $\{e_i\}_{i=1}^n$ of $n$ distinct edges in
$\mathcal{L}$, we have
\begin{eqnarray*}
\textrm{Pr}\{(C_1=0)\cap (C_2=0)\cap \cdots \cap (C_n=0)\}\leq
q^{\frac{n}{4}},
\end{eqnarray*}
where $C_i$ is the indicator of $e_i$ being open,, and $q=1-p$ is
the probability of an edge being closed.
\end{fact}

\subsubsection{Ergodic Theory}
\label{subsubsec:ergodic_theory}

The study object of ergodic theory is the so-called
measure-preserving (m.p.) dynamical system
$(\Omega,~\mathscr{F},~\mu,~T)$, which consists of a set $\Omega$, a
$\sigma$-algebra $\mathscr{F}$ of measurable subsets of $\Omega$, a
nonnegative measure $\mu$ on $(\Omega,~\mathscr{F})$, and an
invertible m.p. transformation $T:~\Omega \rightarrow \Omega$ such
that $\mu(T^{-1}F)=\mu(F)$ for all $F\in \mathscr{F}$. A set $F\in
\mathscr{F}$ is said to be T-invariant if $T^{-1}F=F$. Obviously,
all T-invariant sets in $\mathscr{F}$ form a $\sigma$-algebra.

An m.p. dynamical system $(\Omega,~\mathscr{F},~\mu,~T)$ is said to
be ergodic if the $\sigma$-algebra of T-invariant sets is trivial,
\ie for any invariant set, either it has measure zero or its
complement has measure zero. Another property of the m.p. dynamical
system that implies ergodicity is called mixing: an m.p. dynamical
system $(\Omega,~\mathscr{F},~\mu,~T)$ is said to be mixing if for
all $E,F\in \mathscr{F}$, $\mu (T^n E\cap F)-\mu (E)\mu
(F)\rightarrow 0$ as $n\rightarrow \infty$. For a m.p. dynamical
system which is a product of two m.p. dynamical systems, we have the
following classical result in ergodic theory.

\begin{fact}~\citep[Theorem
2.6.1]{Petersen:Ergodic_Theory} \label{fact:product_ergodicity} \\
The product system of a mixing m.p. dynamical system and an ergodic
m.p. dynamic system is ergodic, that is, for a mixing
$(\Omega,~\mathscr{F},~\mu,~T)$ and an ergodic
$(\Psi,~\mathscr{L},~\nu,~S)$, the product system $(\Omega \times
\Psi,~\mathscr{F} \times \mathscr{L},~\mu \times \nu,~T\times S)$ is
ergodic, where $\mathscr{F} \times \mathscr{L}$ is the
$\sigma$-algebra on $\Omega \times \Psi$ generated by subsets of the
form $F\times L$ $(F\in \mathscr{F}, L\in \mathscr{L})$ and $\mu
\times \nu$ is the corresponding product measure.
\end{fact}

The concepts of ergodicity and mixing can also be defined for a
random model under a probability space $(\Omega,\mathscr{F},\mu)$,
where the m.p. transformation $T$ is replaced by a transformation
group $\{S_x:~x\in \mathbb{R}^d \textrm{ or }\mathbb{Z}^d\}$ indexed
by $\mathbb{R}^d$ or $\mathbb{Z}^d$. For a point process model, the
transformation $S_x$ is usually to shift the realization $\omega \in
\Omega$ by $x$. A random model under a probability space
$(\Omega,\mathscr{F},\mu)$ is said to be ergodic if there exists a
transformation group $\{S_x:~x\in \mathbb{R}^d \textrm{ or
}\mathbb{Z}^d\}$ that acts ergodically on
$(\Omega,\mathscr{F},\mu)$. A transformation group $\{S_x:~x\in
\mathbb{R}^d \textrm{ or }\mathbb{Z}^d\}$ is said to act ergodically
if the $\sigma$-algebra of events invariant under the whole group is
trivial, \ie any invariant event has measure either zero or one.
Moreover, a random model under a probability space
$(\Omega,\mathscr{F},\mu)$ is said to be mixing if there exists a
transformation group $\{S_x:~x\in \mathbb{R}^d \textrm{ or
}\mathbb{Z}^d\}$ such that for all $E,F\in \mathscr{F}$, we have
$\mu (S_x E\cap F)-\mu (E)\mu (F)\rightarrow 0$ as $|x|\rightarrow
\infty$. One direct consequence of an ergodic random model is
presented as below.

\begin{fact} \label{fact:ergodicity}
For an ergodic random model $(\Omega,\mathscr{F},\mu)$, if an event
$E\in \mathscr{F}$ invariant under the whole transformation group
$\{S_x:~x\in \mathbb{R}^d \textrm{ or }\mathbb{Z}^d\}$ occurs wpp.,
\ie $\mu (E)>0$, then it occurs a.s., \ie $\mu (E)=1$.
\end{fact}

\subsection{Proof of Theorem 1}
\label{subsec:proof_thm1}

\subsubsection{Proof of T1.1}
\label{subsubsec:proof_thm11}

To prove T1.1, it suffices to show that for any two given points
$(\lambda_{S1},\lambda_{PT1})$ and $(\lambda_{S2},\lambda_{PT2})$ in
$\Cc$, we can find a path in $\Cc$ that connects these two points.
In particular, the path we constructed is given by a horizontal
segment and a vertical segment as shown in
Fig.~\ref{fig:Proof_Contiguous}, where we assume, without loss of
generality, that $\lambda_{S1} \leq \lambda_{S2}$.

\begin{figure}[htbp]
\centerline{
\begin{psfrags}
\psfrag{P1}[c]{$(\lambda_{S1},\lambda_{PT1})$}
\psfrag{P2}[c]{$(\lambda_{S2},\lambda_{PT2})$}
\psfrag{P}[c]{$(\lambda_{S2},\lambda_{PT1})$}
\psfrag{PT}[c]{$\lambda_{PT}$} \psfrag{S}[c]{$\lambda_S$}
\psfrag{(a)}[c]{(a) $\lambda_{PT1}\leq \lambda_{PT2}$}
\psfrag{(b)}[c]{(b) $\lambda_{PT1}> \lambda_{PT2}$}
\scalefig{0.7}\epsfbox{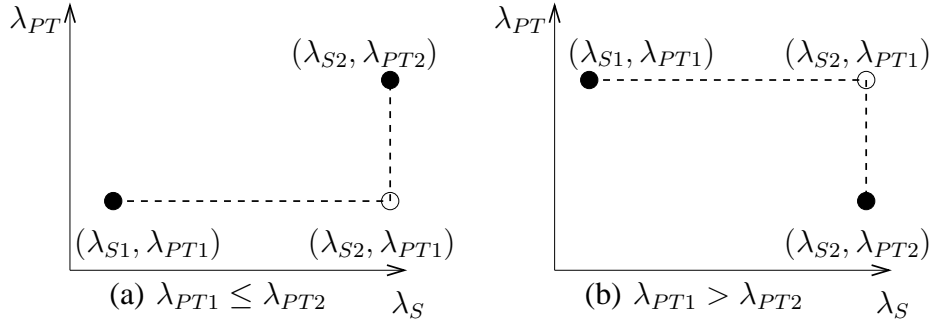}
\end{psfrags}} \caption{The continuous path connecting the two points
$(\lambda_{S1},\lambda_{PT1})$ and $(\lambda_{S2},\lambda_{PT2})$ in
the connectivity region $\Cc$.} \label{fig:Proof_Contiguous}
\end{figure}

Consider case (a) in Fig.~\ref{fig:Proof_Contiguous} where
$\lambda_{PT1}\leq \lambda_{PT2}$. Case (b) can be proven similarly.
First we show every point $(\lambda_S,\lambda_{PT1})$
$(\lambda_{S1}\leq \lambda_S \leq \lambda_{S2})$ on the horizontal
segment belongs to $\Cc$. Let $\lambda'=\lambda_S-\lambda_{S1}$. A
Poisson point process $X$ with density $\lambda_S$ is statistically
equivalent to the superposition of a Poisson point process $X_1$
with density $\lambda_{S1}$ and an independent Poisson point process
$X'$ with density $\lambda'$. It follows that any realization of the
heterogeneous network with densities $\lambda_S$ and $\lambda_{PT1}$
can be generated by adding more secondary nodes to a realization of
the heterogeneous network with densities $\lambda_{S1}$ and
$\lambda_{PT1}$. Thus, the existence of an infinite connected
component in $\Gc (\lambda_{S1},\lambda_{PT1})$ implies the
existence of an infinite connected component in $\Gc
(\lambda_{S},\lambda_{PT1})$. We thus have that
$(\lambda_S,\lambda_{PT1})\in \Cc$ for $(\lambda_{S1}\leq \lambda_S
\leq \lambda_{S2})$.

Now we know that the two end points $(\lambda_{S2},\lambda_{PT1})$
and $(\lambda_{S2},\lambda_{PT2})$ of the vertical segment belong to
$\Cc$. For a point $(\lambda_{S2},\lambda_{PT})$ $(\lambda_{PT1}
\leq \lambda_{PT}\leq \lambda_{PT2})$ on the vertical segment, let
$\lambda'=\lambda_{PT2}-\lambda_{PT}$, then any realization of the
heterogeneous network with densities $\lambda_{S2}$ and
$\lambda_{PT}$ can be obtained by independently removing each
primary transmitter-receiver pair with probability
$\lambda'/\lambda_{PT2}$ from a realization of the heterogeneous
network with densities $\lambda_{S2}$ and $\lambda_{PT2}$. It
follows from the definition of communication link in the secondary
network (see Sec.~\ref{subsubsec:CL}) that the existence of an
infinite connected component in $\Gc (\lambda_{S2},\lambda_{PT2})$
implies the existence of an infinite connected component in $\Gc
(\lambda_S,\lambda_{PT})$. Thus, we have
$(\lambda_{S2},\lambda_{PT})\in \Cc$ $(\lambda_{PT1} \leq
\lambda_{PT}\leq \lambda_{PT2})$.

\subsubsection{Proof of Theorem 1.2}
\label{subsubsec:proof_thm12}

Suppose that $(\lambda_S,\lambda_{PT})\in \Cc$ $(\lambda_{PT}>0)$,
then by using the coupling argument for showing that the vertical
segment belongs to $\Cc$ in the above proof of T1.1, we conclude
that $(\lambda_S,0)\in \Cc$, \ie the $\lambda_S$-axis is the lower
boundary of $\Cc$.

Suppose that $\lambda_{S2}>\lambda_{S1}>0$. In order to prove the
monotonicity of $\lambda^*_{PT}(\lambda_S)$ with $\lambda_S$ it
suffices to show that $\forall~\lambda_{PT}\geq 0$, if
$(\lambda_{S1},\lambda_{PT})\in \Cc$ then
$(\lambda_{S2},\lambda_{PT})\in \Cc$. This is a direct consequence
of the coupling argument for showing that the horizontal segment
belongs to $\Cc$ in the above proof of T1.1.

\subsubsection{Proof of Theorem 1.3}
\label{subsubsec:proof_thm13}

We first establish the ergodicity of the heterogeneous network
model.

\begin{lemma} \label{lemma:ergodicity}
The heterogeneous network model is ergodic.
\end{lemma}

\begin{proof}[Proof of Lemma~\ref{lemma:ergodicity}]
The proof of this lemma is inspired by the proof of the ergodicity
of Poisson Boolean model~\citep[Proposition
2.8]{Meester&Roy:Con_Percolation}. The difficulty here is that for
the heterogeneous network model, we have two correlated Poisson
point processes: the primary transmitters and the primary receivers.
The definition of the shift transformation for the primary network
model is thus more complicated than the standard Poisson Boolean
model with a deterministic radius $\rho$. To prove
Lemma~\ref{lemma:ergodicity}, we first show the ergodicity of the
primary network model, and then we show the mixing property of the
secondary network model. Since the primary network model is
independent of the secondary network model, it follows from
Fact~\ref{fact:product_ergodicity} that the heterogeneous network
model is ergodic.

Let $\mathscr{B}^d$ denote the Borel $\sigma$-algebra in
$\mathbb{R}^d$, and $N$ the set of all simple counting
measures\footnote{A simple counting measure on $\mathscr{B}^d$ is an
integer-valued measure for which the measures of bounded Borel sets
are all finite and the measure of a point is at most 1.} on
$\mathscr{B}^d$. Construct a $\sigma$-algebra $\mathscr{N}$ for $N$
generated by sets of the form
\begin{eqnarray*}
\{n\in N:~n(A)=k\},
\end{eqnarray*}
where $A\in \mathscr{B}^d$ and $k$ is an integer. A point process
$X$ can now be defined as a measurable mapping from a probability
space $(\Omega,~\mathscr{F},~P)$ into
$(N,~\mathscr{N})$~\citep[Chapter
7]{Daley&Jones:Intro_Point_Processes}. The measure $\mu$ on
$\mathscr{N}$ induced by $X$ is defined as $\mu (G)=P(X^{-1}(G))$,
for all $G\in \mathscr{N}$.

In order to define the shift transformation on $\Omega$, it is
convenient to identify $(\Omega,~\mathscr{F})$ with
$(N,~\mathscr{N})$. Let $\omega(A)$ denote the number of points in
$A\in \mathscr{B}^d$, $\forall ~\omega \in \Omega$, and $T_x$ be the
shift according to a vector $x\in \mathbb{R}^d$. Then $T_x$ induces
a shift transformation $S_x:~\Omega \rightarrow \Omega$ through the
equation for every $A\in \mathscr{B}^d$,
\begin{eqnarray} \label{eqn:def_S_x}
(S_x \omega)(A)=\omega (T_x^{-1} A).
\end{eqnarray}

Let $(\Omega_{PT},~\mathscr{F}_{PT},~P_{PT})$ be the probability
space of the Poisson point process $X_{PT}$ for the primary
transmitters with density $\lambda_{PT}$. Let $\Omega_{PR}$ be the
product space $\prod_{n\in \mathbb{N}}\prod_{z\in
\mathbb{Z}^2}C_{R_p}$ for the primary receivers, where
$C_{R_p}=\{(x,y):~x^2+y^2\leq R_p\}$. Then we equip $\Omega_{PR}$
with the usual product $\sigma$-algebra and product measure $P_{PR}$
with all marginal probability measure being $\mu_U$, where $\mu_U$
is a uniform probability measure on $C_{R_p}$. Finally, we set
$\Omega_P=\Omega_{PT}\times \Omega_{PR}$ and equip $\Omega_P$ with
the product measure $P_P=P_{PT}\times P_{PR}$ and the usual product
$\sigma$-algebra. It follows that the primary network model is a
measurable mapping from $\Omega_P$ into $N_{PT}\times \Omega_{PR}$
defined by $(\omega_{PT},~\omega_{PR})\rightarrow
(X_{PT}(\omega_{PT}),~\omega_{PR})$, where $N_{PT}$ is specified in
the definition of the point process.

The positions of the primary transmitters corresponding to
$(\omega_{PT},~\omega_{PR})\in \Omega_{PT}\times \Omega_{PR}$ are
easily known from $\omega_{PT}$. For the primary receivers, the
positions are obtained as follows. Consider binary cubes
\begin{eqnarray*}
K(n,z):=\prod_{i=1}^{2}(z_i 2^{-n},(z_i+1)2^{-n}]~~\textrm{ for all
}n\in \mathbb{N}\textrm{ and }z\in \mathbb{Z}^2.
\end{eqnarray*}
For each primary transmitter $x_{PT}$, there exists a unique
smallest integer $n_0=n_0(x_{PT})$ such that it is contained in a
binary cube $K(n_0,z(n_0,x_{PT}))$ which contains no other primary
transmitters. The relative position of $x_{PT}$'s receiver with
respect to $x_{PT}$ is then given by
$\omega_{PR}(n_0,z(n_0,x_{PT}))$.

Let $e_1$, $e_2$ denote the unit vectors in $\mathbb{R}^2$, then the
translation $T_{e_i}:~\mathbb{R}^2 \rightarrow \mathbb{R}^2$
$(i=1,2)$ defined by $x\rightarrow x+e_i$ induces a shift
transformation $U_{e_i}$ on $\Omega_{PR}$ through the equation
\begin{eqnarray*}
(U_{e_i}\omega_{PR})(n,z)=\omega_{PR} (n,z-2^n e_i), \textrm{ for
$i=1,2$}.
\end{eqnarray*}
Hence $T_{e_i}$ also induces a shift transformation
$\tilde{T}_{e_i}$ on $\Omega_P=\Omega_{PT}\times \Omega_{PR}$ as
follows:
\begin{eqnarray*}
\tilde{T}_{e_i}(\omega_P)=(S_{e_i}\omega_{PT},~U_{e_i}\omega_{PR}),
\textrm{ for $i=1,2$},
\end{eqnarray*}
where $S_{e_i}$ is defined in~(\ref{eqn:def_S_x}). By using
techniques similar to the proof of Boolean models~\citep[Proposition
2.8]{Meester&Roy:Con_Percolation}, we have that the m.p. dynamical
system $(\Omega_P,~\mathscr{F}_P,~P_P,~\tilde{T}_{e_1})$ is ergodic.

Since the transmission range $r_p$ of secondary users is fixed, the
probability space of the secondary network model is the probability
space $(\Omega_S,~\mathscr{F}_S,~P_S)$ for the Poisson point process
$X_S$ of secondary users with density $\lambda_S$. It follows from
the proof of Poisson point processes~\citep[Proposition
2.6]{Meester&Roy:Con_Percolation} that the m.p. dynamical system
$(\Omega_S,~\mathscr{F}_S,~P_S,~S_{e_1})$ is mixing.

Since the primary network model is independent of the secondary
network model, the sample space of the heterogeneous network model
$\Omega$ can be written as the product of $\Omega_P$ and $\Omega_S$,
$\ie$ $\Omega=\Omega_P\times \Omega_S$. We equip $\Omega$ with
product measure $P=P_P \times P_S$ and the usual product
$\sigma$-algebra. Similarly, the translation $T_{e_i}$ ($i=1,2$)
induces a transformation $\hat{T}_{e_i}$ on $\Omega=\Omega_P \times
\Omega_S$, which is given by
\begin{eqnarray*}
\hat{T}_{e_i}(\omega)=(\tilde{T}_{e_i}\omega_P,S_{e_i}\omega_S).
\end{eqnarray*}
Then it follows from Fact~\ref{fact:product_ergodicity} that the
product m.p. dynamical system
$(\Omega,~\mathscr{F},~P,~\hat{T}_{e_1})$ is ergodic. Since the
$\sigma$-algebra invariant under the transformation group
$\{\hat{T}_{z}:~z\in \mathbb{Z}^2\}$ is a subset of the
$\sigma$-algebra invariant under the transformation $\hat{T}_{e_1}$,
we conclude that $\{\hat{T}_{z}:~z\in \mathbb{Z}^2\}$ acts
ergodically, \ie the heterogeneous network model is ergodic.
\end{proof}

Based on Lemma~\ref{lemma:ergodicity}, we have the following lemma.

\begin{lemma} \label{lemma:uniqueness}
The number of infinite connected component in $\Gc
(\lambda_S,\lambda_{PT})$ is a constant a.s., and it can only take
value from $\{0,~1,~\infty\}$.
\end{lemma}

\begin{proof}[Proof of Lemma~\ref{lemma:uniqueness}]
Let $K$ denote the (random) number of infinite connected components
in $\mathcal{G}(\lambda_S,\lambda_{PT})$, then since for all $k\geq
0$, the event $\{K=k\}$ is invariant under the group of shift
transformations, it follows from Lemma~\ref{lemma:ergodicity} and
Fact~\ref{fact:ergodicity} that the event occurs with probability
$0$ or $1$. Consequently, we have that $K$ is an a.s. constant. Then
it suffices to exclude the possibility of $K\geq 2$. This is shown
by contradiction, that is, if there exist $K\geq 2$ infinite
connected components, then they can be linked together as one
connected component wpp. The proof is inspired by the proof of
Proposition 3.3 in~\citep{Meester&Roy:Con_Percolation}, and a major
difference is that here we need to consider the impact of the
primary network on the connectivity of the secondary network.

Suppose that there are $K\geq 2$ infinite connected components a.s.
If we remove all the secondary nodes centered inside a box
$B=[-n,n]^2$, then the resulting secondary network should contain at
least $K$ unbounded components a.s. Let, for $A\subseteq
\mathbb{R}^2$, $\mathcal{G}[A]$ denote the graph formed by secondary
nodes in $A$. Given a box $B$ and $\epsilon
>0$, consider the event
\begin{eqnarray*}
E(B,\epsilon):=\left\{\textrm{$d(U,B)\leq r_p-\epsilon$ for any
infinite connected component $U$ in $\mathcal{G}[B^c]$}\right\}.
\end{eqnarray*}

Partition the box $B$ into squares with side length $a>0$ and let
$\mathcal{S}_a=\{S_1,...,S_N\}$ denote the collection of all the
squares which are adjacent to the boundary of $B$. Clearly, for a
box $B$ and $\epsilon >0$, we can find $a=a(B,\epsilon)\in (0,
r_p/\sqrt{5})$ and $\eta=\eta (a)>0$ such that for any point
$x\notin B$ with $d(x,B)\leq r_p-\epsilon/2$, there exists a square
$S=S(x)\in \mathcal{S}_a$ for which we have $\sup_{y\in S}d(x,y)\leq
r_p-\eta$. This means that, if we center in each square of
$\mathcal{S}_a$ a secondary node and there are neither primary
transmitters nor primary receivers within a bigger box
$\bar{B}=[-n-\max \{r_I,R_I\},n+\max \{r_I,R_I\}]^2$, then every
infinite component $U$ in $\mathcal{G}[B^c]$ with $d(U,B)\leq
r_p-\epsilon$ is connected to some secondary node in
$\mathcal{S}_a$.

Let $E(a,\eta)$ be the event that each square in $\mathcal{S}_a$
contains at least one secondary node and $E(\bar{B})$ the event that
there are neither primary transmitters nor primary receivers within
$\bar{B}$. Since $E(a,\eta)$ depends on the configuration of
secondary nodes inside the box $B$, $E(B,\epsilon)$ depends on the
configuration of secondary nodes outside $B$ and the configuration
of primary nodes, based on the independence of the primary network
and the secondary network, we have
\begin{eqnarray*}
\textrm{Pr}(E(B,\epsilon)\cap E(a,\eta)\cap
E(\bar{B}))=\textrm{Pr}(E(B,\epsilon))\textrm{Pr}(E(a,\eta))\textrm{Pr}(E(\bar{B})|E(B,\epsilon)).
\end{eqnarray*}

If $E(B,\epsilon)$, $E(a,\eta)$ and $E(\bar{B})$ all occur wpp.,
then there is only one infinite connected component\footnote{Since
$a<r_p/\sqrt{5}$, every secondary node in a square of
$\mathcal{S}_a$ is connected to those secondary nodes in the
neighboring squares.} wpp. By using arguments similar to the proof
for Proposition 3.3 in~\citep{Meester&Roy:Con_Percolation}, we have
that there exists a large enough box $B$ and $\epsilon >0$ such that
$\textrm{Pr}(E(B,\epsilon))>0$. Obviously,
$\textrm{Pr}(E(a,\eta))>0$. Moreover, it is easy to see that
$P(E(\bar{B})|E(B,\epsilon))\geq P(E(\bar{B}))>0$.
\end{proof}

Now we have that the number $K$ of infinite connected components is
equal to zero, one or infinity a.s. To exclude the possibility of
$K=\infty$, we can directly apply the proof of Poisson Boolean
models~\citep[Theorem 3.6]{Meester&Roy:Con_Percolation} here, which
is based on several combinatorial results. The details are omitted.

\subsection{Proof of Theorem 2}
\label{subsec:proof_thm2}

\subsubsection{Proof of T2.1}
\label{subsubsec:proof_thm21}

To prove T2.1, it suffices to show that
\begin{itemize}
\item[(a)] for any $\lambda_S\leq \lambda_c(r_p)$, the secondary
network is not connected for any $\lambda_{PT}\geq 0$;
\item[(b)] for any $\lambda_S>\lambda_c(r_p)$, there exists a
$\lambda^*_{PT}(\lambda_S)>0$ such that $\forall~\lambda_{PT}\leq
\lambda^*_{PT}(\lambda_S)$, the secondary network is connected.
\end{itemize}

From Sec.~\ref{subsec:CP}, we know that for a Poisson homogeneous ad
hoc network with density $\lambda$ and transmission range $r$, the
necessary and sufficient condition for connectivity is
$\lambda>\lambda_c(r)$. Since the existence of an infinite connected
component in the secondary network implies the existence of an
infinite connected component in the homogeneous ad hoc network with
the same density and the same transmission range, by using a
coupling argument, we conclude that when $\lambda_S\leq
\lambda_c(r_p)$, there does not exist an infinite connected
component a.s. in the secondary network for any $\lambda_{PT}\geq
0$. This proves part (a).

The basic idea of the proof of part (b) is to approximate the
secondary network $\mathcal{G}(\lambda_S,\lambda_{PT})$ by a
discrete dependent edge-percolation model on the grid. This discrete
dependent edge-percolation model $\mathcal{L}$ is constructed in a
way such that the existence of an infinite connected component in
$\mathcal{L}$ implies the existence of an infinite connected
component in $\mathcal{G}(\lambda_S,\lambda_{PT})$.

Construct the square lattice $\mathcal{L}$ on $\mathbb{R}^2$ with
side length $d$ (see Fig.~\ref{fig:Dual_L}). Note that each site in
$\mathcal{L}$ is virtual and is not related to any node either in
the secondary network or in the primary network. Next we specify the
conditions for an edge being open in $\mathcal{L}$, which is the key
to the mapping from $\mathcal{G}(\lambda_S,\lambda_{PT})$ to
$\mathcal{L}$.

For each edge $e$ in $\mathcal{L}$, let $(x_e,~y_e)$ denote the
middle point of $e$. Then we introduce three random fields $A_e$,
$B_e$, and $C_e$, all associated with the edge $e$ in $\mathcal{L}$,
where $C_e=A_e B_e$ is the indicator of the edge $e$ being open,
$A_e$ represents the condition (C1) of the distance between two
users for the existence of a communication link in the secondary
network, and $B_e$ represents the condition (C2) of the spectrum
opportunity. Specifically, consider the Poisson Boolean model
$\mathcal{B}(X_S,~r_p/2,~\lambda_S)$ where $X_S$ is the Poisson
point process generated by secondary users, then for a horizontal
edge $e$, $A_e=1$ if the following two events (illustrated in
Fig.~\ref{fig:Ae_1_L}) are true:
\begin{itemize}
\item[(i)] there is an occupied L-R crossing of the rectangle $[x_e-3d/4,~x_e+3d/4]\times
[y_e-d/4,~y_e+d/4]$ in $\mathcal{B}(X_S,~r_p/2,~\lambda_S)$;
\item[(ii)] there are two occupied T-B crossings of the square $[x_e-3d/4,~x_e-d/4]\times
[y_e-d/4,~y_e+d/4]$ and the square $[x_e+d/4,~x_e+3d/4]\times
[y_e-d/4,~y_e+d/4]$ in $\mathcal{B}(X_S,~r_p/2,~\lambda_S)$;
\end{itemize}
and $A_e=0$ otherwise. For a vertical edge $e$, the definition of
$A_e$ is similar, where the horizontal and vertical coordinates are
switched in the above two events.

\begin{figure}[htbp]
\centerline{
\begin{psfrags}
\psfrag{xe_ye}[c]{\small{$\mathbf{(x_e,~y_e)}$}}
\psfrag{e}[c]{\small{$\mathbf{e}$}} \psfrag{3d2}[c]{$\frac{3d}{2}$}
\psfrag{d}[c]{$d$} \psfrag{d2}[c]{$\frac{d}{2}$}
\scalefig{0.65}\epsfbox{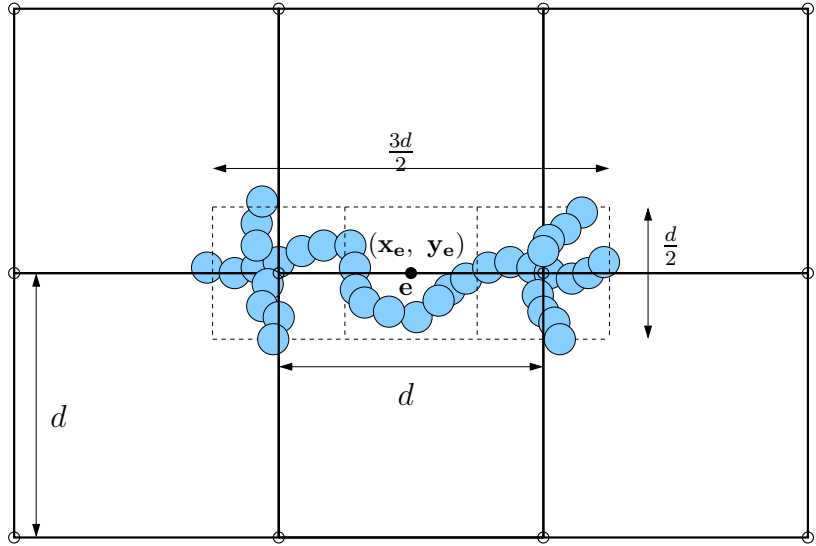}
\end{psfrags}
} \caption{A realization where $A_e=1$ for the edge $e$ (hollow
points are sites in $\mathcal{L}$ and solid segments are edges in
$\mathcal{L}$).} \label{fig:Ae_1_L}
\end{figure}

Next we define the random field $B_e$. For an edge $e$ in
$\mathcal{L}$, $B_e=1$ if $A_e=1$ and the following two events are
true:
\begin{itemize}
\item[(i)] there is no primary transmitter within distance $R_I$ of
any secondary node of the three crossings in the definition of
$A_e$;
\item[(ii)] there is no primary receiver within distance $r_I$ of any
secondary node of the three crossings in the definition of $A_e$;
\end{itemize}
and $B_e=0$ otherwise. It follows from the definition of
communication link in the secondary network (see
Sec.~\ref{subsubsec:CL}) that if $A_e=1$ and $B_e=1$, then the three
crossings in $\mathcal{B}(X_S,~r_p/2,~\lambda_S)$ are also three
crossings in $\mathcal{G}(\lambda_S,\lambda_{PT})$.

Let $C_e=A_e B_e$, then we claim that the edge $e$ is open if
$C_e=1$, and $e$ is closed if $C_e=0$. We observe from
Fig.~\ref{fig:Ae_1_L} that whether the edge $e$ is open is
correlated with the states of the other edges. This model
$\mathcal{L}$ thus is a dependent edge-percolation model.
Furthermore, as shown in Fig.~\ref{fig:Percolation_L}, if there
exists an infinite open connected component in $\mathcal{L}$, then
those crossings associated with the edges in the infinite component
in $\mathcal{L}$ comprise an infinite connected component in
$\mathcal{G}(\lambda_S,\lambda_{PT})$. As a consequence, by
considering the uniqueness of the infinite connected component in
$\mathcal{G}(\lambda_S,\lambda_{PT})$, we only need to prove the
following lemma in order to show T2.1.

\begin{figure}[htbp]
\centerline{\scalefig{0.65}\epsfbox{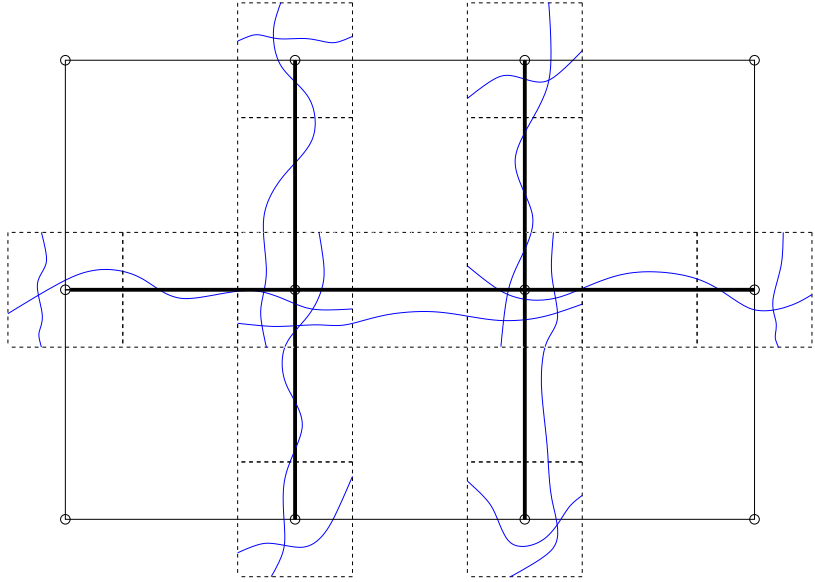}}
\caption{Percolation in $\mathcal{L}$ (thick segments are open edges
in $\mathcal{L}$ and thin segments are closed edges in
$\mathcal{L}$, and blue curves are those crossings associated with
the open edges).} \label{fig:Percolation_L}
\end{figure}

\begin{lemma} \label{lemma:percolation_L}
Let $C(O)$ denote the open connected component containing the origin
$O$ in $\mathcal{L}$. Then given $\lambda_S>\lambda_c(r_p)$,
$\exists~D>0$, $\lambda^*_{PT}>0$ such that for $d=D$ and any
$\lambda_{PT} \leq \lambda^*_{PT}$, we have
\begin{eqnarray*}
\textrm{Pr}\{|C(O)|=\infty\}>0,
\end{eqnarray*}
where $|C(O)|$ is the number of edges in $C(O)$.
\end{lemma}

\begin{proof}[Proof of Lemma~\ref{lemma:percolation_L}]
For an arbitrary edge $e$ in $\mathcal{L}$, let
$q=\textrm{Pr}\{C_e=0\}$, then we have
\begin{eqnarray*}
q=\textrm{Pr}\{(A_e=0)\cup (B_e=0)\} \leq
\textrm{Pr}\{A_e=0\}+\textrm{Pr}\{B_e=0\}.
\end{eqnarray*}

From the result on the crossing probabilities given in
(\ref{eqn:crossing_prob}), we know that when
$\lambda_S>\lambda_c(r_p)$,
\begin{eqnarray*}
\textrm{Pr}\{A_e=0\}&=&\textrm{Pr}\{\textrm{at least one crossing
does not exist}\} \\
&\leq& \left[1-\sigma
\left(\left(\frac{3d}{2},\frac{d}{2}\right),~\lambda_S,~\textrm{L-R}\right)\right]
+\left[1-\sigma
\left(\left(\frac{d}{2},\frac{d}{2}\right),~\lambda_S,~\textrm{T-B}\right)\right]
\\ & &+\left[1-\sigma
\left(\left(\frac{d}{2},\frac{d}{2}\right),~\lambda_S,~\textrm{T-B}\right)\right]
\\
&\rightarrow& 0~~~~\textrm{as}~~~~d\rightarrow \infty,~~~~\ie
\underset{d\rightarrow \infty}{\lim} \textrm{Pr}\{A_e=0\}=0.
\end{eqnarray*}
Thus when $\lambda_S>\lambda_c(r_p)$, $\forall \epsilon
>0$, $\exists~D>0$ such that $\textrm{Pr}\{A_e=0\}<
\frac{\epsilon}{3}$.

Given $A_e=1$, let $S_{R_I}$ be the area of the region covered by
the circles with radii $R_I$ centered at those secondary nodes in
the three crossings, and $S_{r_I}$ be the area of the region covered
by the circles with radii $r_I$ centered at those secondary nodes in
the three crossings. Then we have
\begin{eqnarray*}
\textrm{Pr}\{B_e=0~|~A_e=1\}
&=& \textrm{Pr}\{\exists~\textrm{some primary transmitter in $S_{R_I}$}\} \\
& &+\textrm{Pr}\{\exists~\textrm{some primary receiver in
$S_{r_I}$}\}.
\end{eqnarray*}
Since $S_{R_I}\leq
\left(\frac{3d}{2}+2R_I+r_p\right)\left(\frac{d}{2}+2R_I+r_p\right)$
and $S_{r_I}\leq
\left(\frac{3d}{2}+2r_I+r_p\right)\left(\frac{d}{2}+2r_I+r_p\right)$,
it follows from the basic property of Poisson point processes that
\begin{eqnarray*}
\textrm{Pr}\{B_e=0~|~A_e=1\} &\leq&
1-\exp\left[-\lambda_{PT}\left(\frac{3d}{2}+2R_I+r_p\right)\left(\frac{d}{2}+2R_I+r_p\right)\right]\\
&
&+1-\exp\left[-\lambda_{PT}\left(\frac{3d}{2}+2r_I+r_p\right)\left(\frac{d}{2}+2r_I+r_p\right)\right].
\end{eqnarray*}
Obviously, $\underset{\lambda_{PT} \rightarrow
0}{\lim}~\textrm{Pr}\{B_e=0~|~A_e=1\}=0$ for fixed $d$. Thus if we
choose $d=D$, then $\forall \epsilon >0$, $\exists~\lambda^*_{PT}>0$
such that
\begin{eqnarray*}
\textrm{Pr}\{B_e=0~|~A_e=1\}< \frac{\epsilon}{3}~~\textrm{for all
$\lambda_{PT}\leq \lambda^*_{PT}$.}
\end{eqnarray*}
It implies that when $d=D$, for all $\lambda_{PT}\leq
\lambda^*_{PT}$,
\begin{eqnarray*}
\textrm{Pr}\{B_e=0\}
&=&\textrm{Pr}\{A_e=0\}+\textrm{Pr}\{B_e=0~|~A_e=1\}\textrm{Pr}\{A_e=1\}
\\
&\leq&\textrm{Pr}\{A_e=0\}+\textrm{Pr}\{B_e=0~|~A_e=1\} \\
&<& \frac{2\epsilon}{3}.
\end{eqnarray*}

Thus for $d=D$ and all $\lambda_{PT} \leq \lambda^*_{PT}$, we have
\begin{equation} \label{eqn:small_q}
q\leq \textrm{Pr}\{A_e=0\}+\textrm{Pr}\{B_e=0\}< \epsilon.
\end{equation}

From Fig.~\ref{fig:Ae_1_L}, we can see that if $d\geq \max
\{4R_I+2r_p,~4r_I+2r_p\}$, then the state of edge $e$ is only
correlated with its six adjacent edges and it is independent of
other edges. In this case, by using the `Peierls
argument\footnote{The essence of `Peierls argument' is to make use
of the one-to-one correspondence between a finite open component in
lattice $\mathcal{L}$ containing the origin $O$ and a closed circuit
in the dual lattice $\mathcal{L}^+$ of $\mathcal{L}$ surrounding the
origin $O$.}'~\citep[Chapter 1]{Grimmett:Percolation}, we can show
that if the probability of an edge being closed
$q<\left(\frac{11-2\sqrt{10}}{27}\right)^4$, then
\begin{eqnarray} \label{eqn:per_L}
\textrm{Pr}\{|C(O)|=\infty\}>0.
\end{eqnarray}
The proof of the above statement follows the proof of Theorem 3
in~\citep{Dousse&Etal:05ITN} except that the upper bound on the
probability of $n$ edges all being closed is replaced by the one
given in Fact~\ref{fact:upper_bound_n_closed}.

Thus by combining (\ref{eqn:per_L}) with (\ref{eqn:small_q}), we
conclude that given $\lambda_S>\lambda_c(r_p)$, $\exists~D>0$,
$\lambda^*_{PT}>0$ such that for fixed $d=\max
\{D,4R_I+2r_p,~4r_I+2r_p\}$ and any $\lambda_{PT} \leq
\lambda^*_{PT}$,
\begin{eqnarray*}
\textrm{Pr}\{|C(O)|=\infty\}>0.
\end{eqnarray*}
Notice that $\lambda^*_{PT}$ depends on $D$ which is chosen
according to the crossing probability and is determined by
$\lambda_S$. As a consequence, $\lambda^*_{PT}$ is a function of
$\lambda_S$, $\ie$ $\lambda^*_{PT}=\lambda^*_{PT}(\lambda_S)$.
\end{proof}

\subsubsection{Proof of T2.2}
\label{subsubsec:proof_thm22}

From the conditions for the existence of a communication link in the
secondary network specified in Sec.~\ref{subsubsec:CL}, we know that
for every secondary node in an infinite connected component, there
can exist neither any primary transmitter within distance $R_I$ of
it nor any primary receiver within distance $r_I$ of it. In other
words, every secondary node in an infinite connected component must
be located outside all the circles centered at the primary
transmitters and the primary receivers with radii $R_I$ and $r_I$,
respectively. Thus, if there is an infinite connected component in
the secondary network, then an infinite vacant component must exist
in the two Poisson Boolean models
$\mathcal{B}(X_{PT},~R_{PT},~\lambda_{PT})$ and
$\mathcal{B}(X_{PR},~R_{PR},~\lambda_{PT})$ driven by the primary
transmitters and the primary receivers, respectively. Here $R_{PT}$
and $R_{PR}$ are some appropriate radii which will be specified
later. A natural choice for $R_{PT}$ is $R_I$, but if we consider
the counterexample given in Fig.~\ref{fig:Counterexample}, then we
can clearly see that even if there is an infinite path in the
secondary network, no infinite vacant component exists in the
Poisson Boolean model $\mathcal{B}(X_{PT},~R_I,~p\lambda_P)$ driven
by the primary transmitters. Similarly, counterexamples can be
easily constructed for choosing $R_{PR}=r_I$.

\begin{figure}[htbp]
\centerline{
\begin{psfrags}
\psfrag{RI}[c]{$R_I$} \psfrag{O}[c]{$O$}
\psfrag{symbol}[c]{$\mathbf{\cdots \cdots}$} \psfrag{Pri}[l]{Primary
Transmitter} \psfrag{Sec}[l]{Secondary User}
\scalefig{0.4}\epsfbox{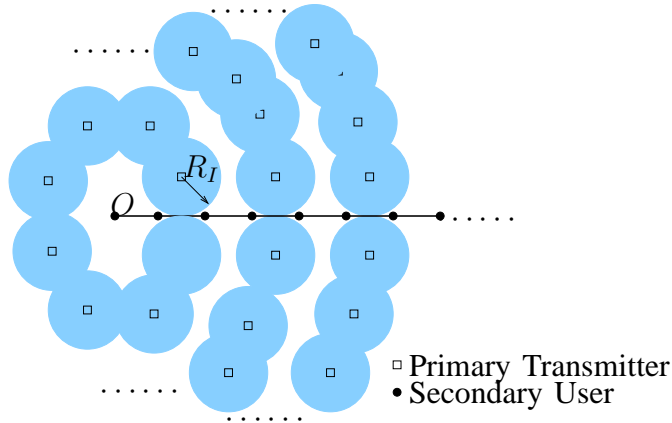}
\end{psfrags}
} \caption{A counterexample for choosing $R_{PT}=R_I$. All the
secondary nodes in the infinite path are located outside those
circles centered at the primary transmitters with radii $R_I$, which
form a series of rings surrounding the origin $O$, and there is no
infinite vacant component in the Poisson Boolean model
$\mathcal{B}(X_{PT},~R_I,~\lambda_{PT})$ driven by the primary
transmitters.} \label{fig:Counterexample}
\end{figure}

Suppose there is an infinite connected component in the secondary
network. Then we can find a sequence of secondary users
$\{S_1,~S_2,~S_3,~\cdots\}$ such that they comprise an infinite path
starting from $S_1$ (see Fig.~\ref{fig:Inf_Path}).

\begin{figure}[htbp]
\centerline{
\begin{psfrags}
\psfrag{S1}[c]{$S_1$} \psfrag{S2}[c]{$S_2$} \psfrag{S3}[c]{$S_3$}
\psfrag{S4}[c]{$S_4$} \psfrag{S5}[c]{$S_5$} \psfrag{S6}[c]{$S_6$}
\psfrag{S7}[c]{$S_7$} \psfrag{S8}[c]{$S_8$}
\psfrag{symbol}[c]{$\mathbf{\cdots \cdots}$}
\scalefig{0.65}\epsfbox{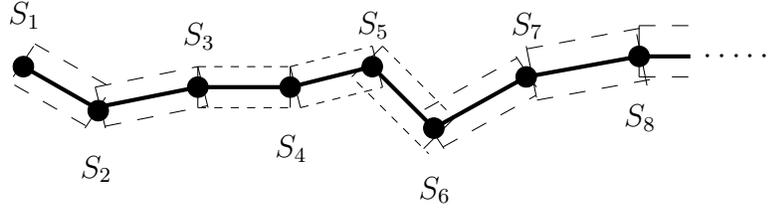}
\end{psfrags}
} \caption{An infinite path in the secondary network. The dashed
segments form an inner bound on the infinite vacant component in the
Poisson Boolean model driven by the primary receivers.}
\label{fig:Inf_Path}
\end{figure}

Assume that $S_i$ and $S_{i+1}$ $(i\geq 1)$ are two adjacent
secondary nodes in the above infinite path. Notice that the distance
$d_{i,~i+1}$ between $S_i$ and $S_{i+1}$ satisfies $d_{i,~i+1}\leq
r_p<r_I$, where the second inequality $r_p<r_I$ follows from the
fact that the minimum transmission power for successful reception is
in general higher than the maximum allowable interference power.

As we know, all the primary receivers must be outside the two
circles with radii $r_I$ centered at $S_i$ and $S_{i+1}$,
respectively, as shown in Fig.~\ref{fig:Nece_Proof}. Given
$\epsilon>0$, consider the rectangle
$\left[-\frac{d_{i,i+1}}{2},~\frac{d_{i,i+1}}{2}\right]\times
[-\epsilon,~\epsilon]$ between $S_i$ and $S_{i+1}$. By a simple
computation in geometry, we have that the minimum distance from all
the primary receivers to the rectangle is
$\sqrt{r_I^2-\frac{r_p^2}{4}}-\epsilon$. As illustrated in
Fig.~\ref{fig:Inf_Path}, it implies that there exists an infinite
vacant component in the Poisson Boolean model
$\mathcal{B}\left(X_{PR},~\sqrt{r_I^2-\frac{r_p^2}{4}}-\epsilon,~\lambda_{PT}\right)$
driven by the primary receivers\footnote{This technique used here
can also be applied to the case when $r_p\geq r_I$, where only the
minimum distance from all the primary receiver to the bar between
$S_i$ and $S_{i+1}$ needs to be recomputed.}. By recalling the known
results in Sec.~\ref{subsubsec:STTD}, we thus conclude that for all
$\epsilon>0$,
\begin{eqnarray*}
\lambda_{PT} \leq \left(2\sqrt{r_I^2-r_p^2/4}-\epsilon \right)^{-2}
\lambda_c(1).
\end{eqnarray*}
Let $\epsilon \rightarrow 0$, then it yields
\begin{eqnarray*}
\lambda_{PT} \leq \frac{1}{4r_I^2-r_p^2}\lambda_c(1).
\end{eqnarray*}

\begin{figure}[htbp]
\centerline{
\begin{psfrags}
\psfrag{S1}[c]{$S_i$} \psfrag{S2}[c]{$S_{i+1}$}
\psfrag{rI}[c]{$r_I$} \psfrag{e}[c]{$\epsilon$}
\psfrag{d}[c]{$d_{i,~i+1}$}
\scalefig{0.4}\epsfbox{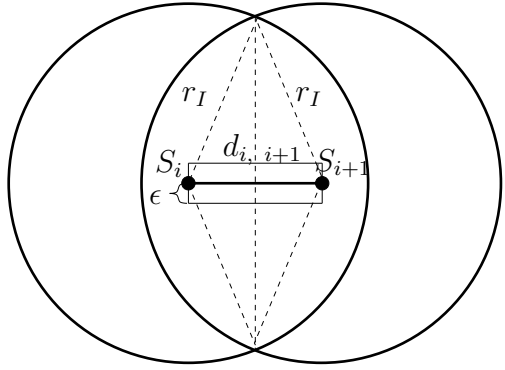}
\end{psfrags}
} \caption{One edge $(S_i,~S_{i+1})$ in the infinite path.}
\label{fig:Nece_Proof}
\end{figure}

The other term $\frac{1}{4R_I^2-r_p^2}\lambda_c(1)$ in the upper
bound is obtained by applying the same argument to the Poisson
Boolean model driven by the primary transmitters.

\subsection{Proof of Theorem 3}
\label{subsec:proof_thm3}

Consider the connected component $C_A$ containing an arbitrarily
chosen secondary user $A$. Assuming that $|C_A|>1$, we construct a
branching process as follows. Notice that if $|C_A|>1$ where $|C_A|$
is the number of users contained in $C_A$, then $A$ must see the
opportunity, \ie $\overline{\mathbb{I}(A,r_I,\textrm{rx})}\cap
\overline{\mathbb{I}(A,R_I,\textrm{tx})}$ is true. Call $A$ the
initial point (or $0$-th generation) of the branching process. Then
the children of $A$ (\ie the $1$st generation of the branching
process) are secondary users which satisfy the following two
conditions:
\begin{itemize}
\item[(i)] it is within distance $r_p$ of $A$, where $r_p$ is the
transmission range of secondary users;
\item[(ii)] there exist neither any primary receiver within
distance $r_I$ of the secondary user nor any primary transmitter
within distance $R_I$ of the secondary user.
\end{itemize}

The $n$-th ($n\geq 2$) generation of the branching process are
obtained similarly, and they are connected to their parents in the
$(n-1)$-th generation of the branching process via bidirectional
links. Obviously, all the secondary users in $C_A$ are counted in
the constructed branching process model. But some of them may
probably be counted more than once, since we do not exclude the
previous $n$ generations (including generation $0$) when we consider
the $n$-th generation. Thus, this branching process gives us an
upper bound on the number of secondary users in $C_A$. It follows
that if the branching process does not grow to infinity wpp., then
there does not exist an infinite connected component a.s. in
$\mathcal{G}(\lambda_S,\lambda_{PT})$, due to the stationarity of
the heterogeneous network model. Since the conditional average
degree $\mu$ is the average number of offspring for every
generation, the necessary condition follows immediately from the
classic theorem for branching processes~\citep[Theorem
2.1.1]{Franceschetti&Meester:Random_Network_Comm}.

\subsection{Proof of Theorem 4}
\label{subsec:proof_thm4}

From the construction of the dependent site-percolation model
$\mathcal{L}$, we know that the existence of an infinite occupied
component in $\mathcal{L}$ implies the existence of an infinite
connected component in $\mathcal{G}(\lambda_S,\lambda_{PT})$. Then
in order to obtain a sufficient condition for the connectivity of
the secondary network, it suffices to find a sufficient condition
for the existence of an infinite occupied component in
$\mathcal{L}$.

Let $p$ be the probability that one site is occupied. Then based on
the definition of the upper critical probability $p_c$ of
$\mathcal{L}$, we have that if $p>p_c$, an infinite occupied
component containing the origin exists in $\mathcal{L}$ wpp. It
implies that if $p>p_c$, there exists an infinite connected
component in the secondary network wpp. Since the event that there
exists an infinite connected component in the secondary network is
invariant under the group of shift transformations, it follows from
the ergodicity of the heterogeneous network model (see
Lemma~\ref{lemma:ergodicity}) that if $p>p_c$, there exists an
infinite connected component in the secondary network a.s.

Based on the definition of occupied site in $\Lc$, we have
\begin{eqnarray*}
p&=&[1-\exp (-\lambda_S
d^2)]\textrm{Pr}\left\{\overline{\mathbb{I}\left(A,r_I,\textrm{rx}\right)}\cap
\overline{\mathbb{I}\left(A,R_I,\textrm{tx}\right)}\right\}
\\
&=& \left[1-\exp \left(-\frac{\lambda_S r_p^2}{8}\right)\right]\exp
\left\{-\lambda_{PT} \pi \left[R_I^2+r_I^2
-I(R_I,R_p,r_I)\right]\right\}.
\end{eqnarray*}
In the last step,
$\textrm{Pr}\left\{\overline{\mathbb{I}\left(A,r_I,\textrm{rx}\right)}\cap
\overline{\mathbb{I}\left(A,R_I,\textrm{tx}\right)}\right\}$ has
been obtained by setting the distance $d=0$ in the expression for
the probability of a unidirectional opportunity between two
secondary users with distance $d$ apart given in Proposition 1
in~\citep{Ren&etal:08JSAC}.


\section{Conclusion and Future Directions}
\label{sec:conclusion}

We have studied the connectivity of a large-scale ad hoc
heterogeneous wireless network in terms of the occurrence of the
percolation phenomenon. We have introduced the concept of
connectivity region to specify the dependency of connectivity on the
density of the secondary users and the traffic load of the primary
users. We have shown several basic properties of the connectivity
region: the contiguity of the region, the monotonicity of the
boundary, and the uniqueness of the infinite connected component. We
have analytically characterized the critical density of the
secondary users and the critical density of the primary
transmitters; they jointly specify the profile of the connectivity
region. We have also established a necessary and a sufficient
condition for connectivity, which give an outer and an inner bound,
respectively, on the connectivity region. Furthermore, by examining
the impacts of the secondary users' transmission power on the
connectivity region and on the conditional average degree of a
secondary user, we have demonstrated the tradeoff between proximity
and spectrum opportunity. In establishing these results, we have
used techniques and theories in continuum percolation, including the
coupling argument, ergodic theory, the discretization technique, and
the approximation using a branching process.

To highlight unique design tradeoffs in heterogeneous networks, we
have ignored the fading effect and the mutual interference between
secondary users. If we take into account these factors, then the
received signal to interference-plus-noise ratios at two secondary
users will replace the distance between them in the condition (C1)
for the existence of a communication link between them. This will
result in a random connection model with correlated links, where the
correlation between links is due to the mutual interference and the
condition (C2) on the presence of the bidirectional opportunity.
Although the connectivity region can still be defined in the same
way, there will be another tradeoff between proximity and mutual
interference besides the tradeoff between proximity and opportunity.
The combination of these two tradeoffs will significantly complicate
the characterization of the connectivity of the secondary network.
We hope results obtained in this paper serve as a first step toward
solving this more complex problem.


\section*{Acknowledgments}
The authors would like to thank Drasvin Cheowtirakul for his help in
generating several simulation results.


\section*{Appendix A: Expression for Conditional Average Degree}

\renewcommand{\theequation}{A\arabic{equation}}
\setcounter{equation}{0}

The expression for the conditional average degree $\mu$ of a
secondary user is presented in the following proposition.
\begin{proposition} \label{pro:cond_avg_deg}
Let $\lambda_S$ and $\lambda_{PT}$ be the density of secondary users
and primary transmitters, respectively. Let $r_I$ and $R_I$ be the
interference range of the secondary and primary users, respectively,
and $r_p$ and $R_p$ the transmission range of the secondary and
primary users, respectively. Then the conditional average degree
$\mu$ of a secondary user is given by
\begin{eqnarray} \label{eqn:cond_avg_deg}
\mu &=& \left(\lambda_S \pi r_p^2\right)\cdot
g(\lambda_{PT},r_p,r_I,R_p,R_I) \nn \\
    &=& \lambda_S \pi r_p^2
\int_0^{r_p}~\frac{2t}{r_p^2}\exp\Bigg\{-\lambda_{PT}\Big[\pi
(r_I^2+R_I^2+I(R_I,R_p,r_I))-S_I (t,r_I,r_I)-S_I
(t,R_I,R_I) \nonumber \\
&
&~~~~~~~~~~~~~~~~~~~~~~~~~~~~~~~~~~-\underset{\Sc_{U2}(t,R_I,R_I)}{\iint}\frac{S_{I2}(r,\theta,R_p,t,r_I)}{\pi
R_p^2}r\mathrm{d}r\mathrm{d}\theta \Big]\Bigg\}\mathrm{d}t,
\end{eqnarray}
where
\begin{eqnarray*}
I(R_I,R_p,r_I)=2 \int_{0}^{R_I} t\frac{S_I(t,R_p,r_I)}{\pi
R_p^2}\mathrm{d}t,
\end{eqnarray*}
$S_I (t,r_1,r_2)$ the common area of two circles with radii $r_1$
and $r_2$ and centered $t$ apart (see
Fig.~\ref{fig:Illustration_S}(a)), and $\Sc_{U2}(t,r_1,r_2)$ is the
union of two circles with radii $r_1$ and $r_2$ and centered $t$
apart (see Fig.~\ref{fig:Illustration_S}(b)).
$S_{I2}(r,\theta,R_p,t,r_I)$ is the intersection area between one
circle with radius $R_p$ and the union of the two circles with both
radii $r_I$ (see Fig.~\ref{fig:Illustration_S}(c)). For
$S_{I2}(r,\theta,R_p,t,r_I)$, the two identical circles are centered
$t$ apart, and the other circle is centered at $(r,\theta)$, where
the middle point of the centers of the two identical circles is
chosen to be the origin $O$.
\end{proposition}

\begin{figure}[htbp]
\centerline{
\begin{psfrags}
\psfrag{R_p}[c]{$R_p$} \psfrag{r_1}[c]{$r_1$} \psfrag{r_2}[c]{$r_2$}
\psfrag{r_I}[c]{$r_I$} \psfrag{r}[c]{$r$}
\psfrag{theta}[c]{$\theta$} \psfrag{t}[c]{$t$} \psfrag{O}[c]{$O$}
\psfrag{S_I}[c]{$S_I (t,r_1,r_2)$} \psfrag{S_U2}[c]{$\Sc_{U2} (t,
r_1,r_2)$} \psfrag{S_I2}[c]{$S_{I2} (r,\theta,R_p,t,r_I)$}
\psfrag{(a)}[c]{(a)} \psfrag{(b)}[c]{(b)} \psfrag{(c)}[c]{(c)}
\scalefig{0.65}\epsfbox{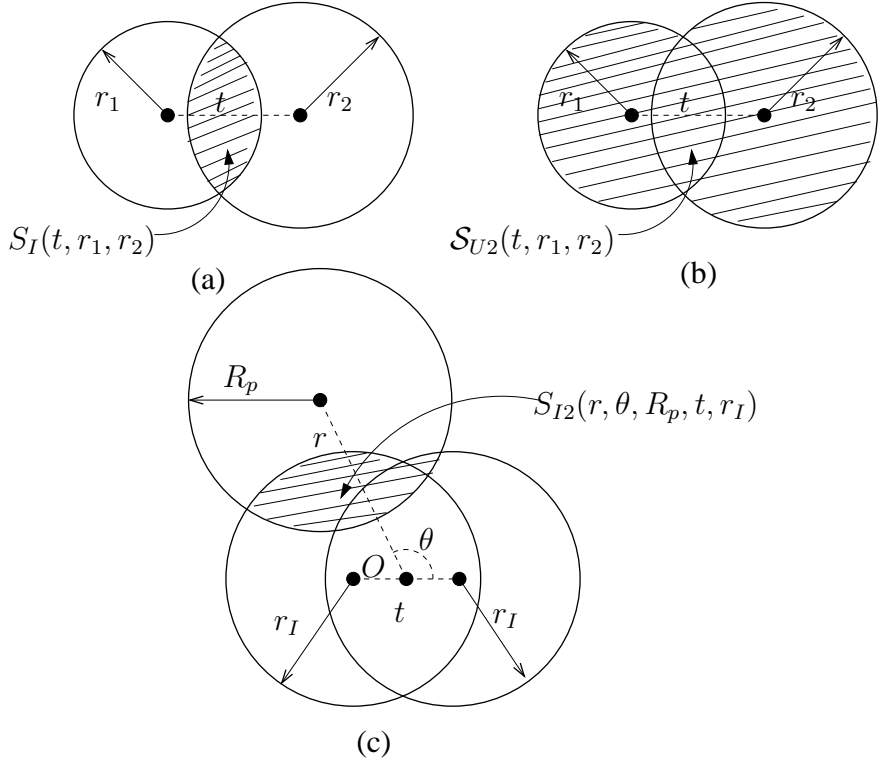}
\end{psfrags}
} \caption{An illustration of $S_I (t,r_1,r_2)$ (the common area of
two circles with radii $r_1$ and $r_2$ and centered $t$ apart),
$\Sc_{U2}(t,r_1,r_2)$ (the union area of two circles with radii
$r_1$ and $r_2$ and centered $t$ apart), and
$S_{I2}(r,\theta,R_p,t,r_I)$ (the intersection area between one
circle with radius $R_p$ and the union of the two identical circles
with radii $r_I$). } \label{fig:Illustration_S}
\end{figure}

The expressions for $I(R_I,R_p,r_I)$ and $S_I(t,r_1,r_2)$ can be
obtained in explicit form, which can be found
in~\citep[Appendix~A]{Ren&etal:08JSAC}. The expression for
$S_{I2}(r,\theta,R_p,t,r_I)$ depends on the expression for the
common area of three circles which is tedious and is given
in~\citep{Fewell:Overlap_Three_Cir}. By applying the basic property
of the exponential function to (\ref{eqn:cond_avg_deg}), we can
easily show that $g(\cdot)$ is a strictly decreasing function of
$\lambda_{PT}$.

\begin{proof}
Let $\mathbb{K}_S(A)$ denote the event that there exist exactly $k$
neighbors of a secondary user $A$. We thus have
\begin{eqnarray*} \label{eqn:mu1}
\mu &=&
\mathbb{E}[deg(A)|~\overline{\mathbb{I}(A,r_I,\textrm{rx})}\cap
\overline{\mathbb{I}(A,R_I,\textrm{tx})}] \\
&=& \mathbb{E}_K
[\mathbb{E}[deg(A)|~\overline{\mathbb{I}(A,r_I,\textrm{rx})}\cap
\overline{\mathbb{I}(A,R_I,\textrm{tx})}\cap \mathbb{K}_S (A)]] \\
&=& \sum_{k=0}^{\infty} e^{-\lambda_S \pi
r_p^2}\frac{\left(\lambda_S \pi
r_p^2\right)^k}{k!}\mathbb{E}[deg(A)|~\overline{\mathbb{I}(A,r_I,\textrm{rx})}\cap
\overline{\mathbb{I}(A,R_I,\textrm{tx})}\cap \mathbb{K}_S (A)].
\end{eqnarray*}

When $k=0$, it is obvious that $deg(A)=0$. When $k>0$, let $B_i$ be
a neighbor of $A$, and $\mathbf{1}_{Bi}$ an indicator function for
$B_i$ such that $\mathbf{1}_{Bi}=1$ if
$\overline{\mathbb{I}(B_i,r_I,\textrm{rx})}\cap
\overline{\mathbb{I}(B_i,R_I,\textrm{tx})}$ occurs and
$\mathbf{1}_{Bi}=0$ otherwise. Then by considering the statistical
independence and equivalence of the $k$ secondary users, we have
\begin{eqnarray*}
&
&~~~\mathbb{E}[deg(A)|~\overline{\mathbb{I}(A,r_I,\textrm{rx})}\cap
\overline{\mathbb{I}(A,R_I,\textrm{tx})}\cap \mathbb{K}_S (A)] \\
& &= \sum_{i=1}^k
\mathbb{E}[\mathbf{1}_{Bi}|~\overline{\mathbb{I}(A,r_I,\textrm{rx})}\cap
\overline{\mathbb{I}(A,R_I,\textrm{tx})}]  \\
& &= k
\mathbb{E}[\mathbf{1}_{B1}|~\overline{\mathbb{I}(A,r_I,\textrm{rx})}\cap
\overline{\mathbb{I}(A,R_I,\textrm{tx})}]  \\
& &= k \textrm{Pr}\{\overline{\mathbb{I}(B_1,r_I,\textrm{rx})}\cap
\overline{\mathbb{I}(B_1,R_I,\textrm{tx})}|~\overline{\mathbb{I}(A,r_I,\textrm{rx})}\cap
\overline{\mathbb{I}(A,R_I,\textrm{tx})}\}  \\
& &= k
\frac{\textrm{Pr}\{\overline{\mathbb{I}(B_1,r_I,\textrm{rx})}\cap
\overline{\mathbb{I}(B_1,R_I,\textrm{tx})}\cap
\overline{\mathbb{I}(A,r_I,\textrm{rx})}\cap
\overline{\mathbb{I}(A,R_I,\textrm{tx})}\}}{\textrm{Pr}\{\overline{\mathbb{I}(A,r_I,\textrm{rx})}\cap
\overline{\mathbb{I}(A,R_I,\textrm{tx})}\}}
\end{eqnarray*}
It follows that
\begin{eqnarray} \label{eqn:mu2}
\mu = \lambda_S \pi r_p^2
\frac{\textrm{Pr}\{\overline{\mathbb{I}(B_1,r_I,\textrm{rx})}\cap
\overline{\mathbb{I}(B_1,R_I,\textrm{tx})}\cap
\overline{\mathbb{I}(A,r_I,\textrm{rx})}\cap
\overline{\mathbb{I}(A,R_I,\textrm{tx})}\}}{\textrm{Pr}\{\overline{\mathbb{I}(A,r_I,\textrm{rx})}\cap
\overline{\mathbb{I}(A,R_I,\textrm{tx})}\}}.
\end{eqnarray}

According to the definition of spectrum opportunity,
$\textrm{Pr}\{\overline{\mathbb{I}(A,r_I,\textrm{rx})}\cap
\overline{\mathbb{I}(A,R_I,\textrm{tx})}\}$ can be obtained by
setting the distance $d=0$ in the expression for the probability of
a unidirectional opportunity between two secondary users with
distance d apart given in Proposition 1 in~\citep{Ren&etal:08JSAC}:
\begin{eqnarray} \label{eqn:mu3}
\textrm{Pr}\{\overline{\mathbb{I}(A,r_I,\textrm{rx})}\cap
\overline{\mathbb{I}(A,R_I,\textrm{tx})}\}=\exp [-\lambda_{PT}
\pi(r_I^2+R_I^2-I(R_I,R_p,r_I))].
\end{eqnarray}

Next we derive the expression for the probability of a bidirectional
opportunity, \ie
$\textrm{Pr}\{\overline{\mathbb{I}(B_1,r_I,\textrm{rx})}\cap
\overline{\mathbb{I}(B_1,R_I,\textrm{tx})}\cap
\overline{\mathbb{I}(A,r_I,\textrm{rx})}\cap
\overline{\mathbb{I}(A,R_I,\textrm{tx})}\}$, which depends on the
location of $B_1$ only through its distance to $A$. Since $B_1$ is
uniformly distributed within distance $r_p$ of $A$, the density
function of the distance $t$ between $B_1$ and $A$ is given by
$\frac{2t}{r_p^2}$ for $0\leq t\leq r_p$. In this case, the
probability of a bidirectional opportunity can be written as
\begin{eqnarray} \label{eqn:mu4}
& &\textrm{Pr}\{\overline{\mathbb{I}(B_1,r_I,\textrm{rx})}\cap
\overline{\mathbb{I}(B_1,R_I,\textrm{tx})}\cap
\overline{\mathbb{I}(A,r_I,\textrm{rx})}\cap
\overline{\mathbb{I}(A,R_I,\textrm{tx})}\} \nn \\
&=&\int_0^{r_p}
\frac{2t}{r_p^2}\textrm{Pr}\{\overline{\mathbb{I}(B_1,r_I,\textrm{rx})}\cap
\overline{\mathbb{I}(B_1,R_I,\textrm{tx})}\cap
\overline{\mathbb{I}(A,r_I,\textrm{rx})}\cap
\overline{\mathbb{I}(A,R_I,\textrm{tx})}|~d(B_1,A)=t\}\mathrm{d}t,~~~~
\end{eqnarray}
where the integrand can be written as
\begin{eqnarray} \label{eqn:mu5}
& &\textrm{Pr}\{\overline{\mathbb{I}(B_1,r_I,\textrm{rx})}\cap
\overline{\mathbb{I}(B_1,R_I,\textrm{tx})}\cap
\overline{\mathbb{I}(A,r_I,\textrm{rx})}\cap
\overline{\mathbb{I}(A,R_I,\textrm{tx})}|~d(B_1,A)=t\} \nn \\
&=&\textrm{Pr}\{\overline{\mathbb{I}(B_1,R_I,\textrm{tx})}\cap
\overline{\mathbb{I}(A,R_I,\textrm{tx})}|~\overline{\mathbb{I}(B_1,r_I,\textrm{rx})}\cap
\overline{\mathbb{I}(A,r_I,\textrm{rx})}\cap d(B_1,A)=t\} \nn \\
& &\textrm{Pr}\{\overline{\mathbb{I}(B_1,r_I,\textrm{rx})}\cap
\overline{\mathbb{I}(A,r_I,\textrm{rx})}|~d(B_1,A)=t\}.
\end{eqnarray}

Next, we compute the two probabilities in (\ref{eqn:mu5}) one by
one. Since the primary receivers admit a Poisson point process with
density $\lambda_{PT}$, we have
\begin{eqnarray} \label{eqn:mu6}
\textrm{Pr}\{\overline{\mathbb{I}(B_1,r_I,\textrm{rx})}\cap
\overline{\mathbb{I}(A,r_I,\textrm{rx})}|~d(B_1,A)=t\}=\exp
[-\lambda_{PT}(2\pi r_I^2-S_I (t,r_I,r_I))],
\end{eqnarray}
where $S_I (t,r_I,r_I)$ is the common area of two circles with both
radii $r_I$ and centered $t$ apart (see
Fig.~\ref{fig:Illustration_S}(a)).

Let $X_{PT}$ denote the Poisson point process formed by primary
transmitters. If we remove from $X_{PT}$ primary transmitters whose
receivers are within distance $r_I$ of $B_1$ or $A$, then it follows
from Coloring Theorem~\citep[Chapter 5]{Kingman:Poisson} that all
the remaining primary transmitters form another Poisson point
process with density $\lambda_{PT} \left[1-\frac{S_{I2}
(r,\theta,R_p,t,r_I)}{\pi R_p^2}\right]$, where $S_{I2}
(r,\theta,R_p,t,r_I)$ is the area of the circle with radius $R_p$
and centered at $(r,\theta)$ intersecting the two circles with both
radii $r_I$ and centered $t$ apart (see
Fig.~\ref{fig:Illustration_S}(c)). We thus have
\begin{eqnarray} \label{eqn:mu7}
& &\textrm{Pr}\{\overline{\mathbb{I}(B_1,R_I,\textrm{tx})}\cap
\overline{\mathbb{I}(A,R_I,\textrm{tx})}|~\overline{\mathbb{I}(B_1,r_I,\textrm{rx})}\cap
\overline{\mathbb{I}(A,r_I,\textrm{rx})}\cap d(B_1,A)=t\} \nn \\
&=& \exp
\left\{-\lambda_{PT}\underset{\Sc_{U2}(t,R_I,R_I)}{\iint}\left[1-\frac{S_{I2}
(r,\theta,R_p,t,r_I)}{\pi
R_p^2}r\mathrm{d}r\mathrm{d}\theta\right]\right\} \nn \\
&=& \exp \left\{-\lambda_{PT}\left[2\pi R_I^2-S_I
(t,R_I,R_I)-\underset{\Sc_{U2}(t,R_I,R_I)}{\iint}\frac{S_{I2}
(r,\theta,R_p,t,r_I)}{\pi
R_p^2}r\mathrm{d}r\mathrm{d}\theta\right]\right\},
\end{eqnarray}
where $\Sc_{U2}(t,R_I,R_I)$ is the union of two circles with both
radii $R_I$ and centered $t$ apart (see
Fig.~\ref{fig:Illustration_S}(b)).

Substitute (\ref{eqn:mu6},~\ref{eqn:mu7}) into (\ref{eqn:mu5}), we
have
\begin{eqnarray} \label{eqn:mu8}
& &\textrm{Pr}\{\overline{\mathbb{I}(B_1,r_I,\textrm{rx})}\cap
\overline{\mathbb{I}(B_1,R_I,\textrm{tx})}\cap
\overline{\mathbb{I}(A,r_I,\textrm{rx})}\cap
\overline{\mathbb{I}(A,R_I,\textrm{tx})}|~d(B_1,A)=t\} \nn \\
&=& \exp \Bigg\{-\lambda_{PT}\Big[2\pi (r_I^2+R_I^2)-S_I
(t,r_I,r_I)-S_I (t,R_I,R_I) \nn \\
&
&~~~~~~~~~~~~~~~~-\underset{\Sc_{U2}(t,R_I,R_I)}{\iint}\frac{S_{I2}
(r,\theta,R_p,t,r_I)}{\pi R_p^2}r\mathrm{d}r\mathrm{d}\theta \Big]
\Bigg\}.
\end{eqnarray}
The expression for the conditional average degree $\mu$ thus follows
by plugging (\ref{eqn:mu8}) into (\ref{eqn:mu4}) and then
(\ref{eqn:mu3},~\ref{eqn:mu4}) into (\ref{eqn:mu2}).
\end{proof}


\section*{Appendix B: Proof of Corollary 3}

\renewcommand{\theequation}{B\arabic{equation}}
\setcounter{equation}{0}

From~\citep[Appendix~A]{Ren&etal:08JSAC}  and
Fig.~\ref{fig:Illustration_S}(b,~c), we know that when $r_I\geq
R_p+R_I$,
\begin{eqnarray} \label{eqn:int_I1_AE}
I(R_I,R_p,r_I) &=& R_I^2, \\
\label{eqn:int_S_U2}
\underset{\Sc_{U2}(t,R_I,R_I)}{\iint}\frac{S_{I2}(r,\theta,R_p,t,r_I)}{\pi
R_p^2}r\mathrm{d}r\mathrm{d}\theta &=& \Sc_{U2}(t,R_I,R_I) = 2\pi
R_I^2-S_I (t,R_I,R_I).
\end{eqnarray}

Substitute (\ref{eqn:int_I1_AE},~\ref{eqn:int_S_U2}) into
(\ref{eqn:cond_avg_deg}), we have
\begin{eqnarray} \label{eqn:mu_AE1}
\mu = \lambda_S \pi r_p^2 \int_0^{r_p}~\frac{2t}{r_p^2}\exp
[-\lambda_{PT}(\pi r_I^2-S_I (t,r_I,r_I))]\mathrm{d}t.
\end{eqnarray}
Plugging the expression for
$S_I(t,r_I,r_I)$~\citep[Appendix~A]{Ren&etal:08JSAC} into
(\ref{eqn:mu_AE1}) yields
\begin{eqnarray*}
\mu = \lambda_S \pi r_p^2 \int_0^{r_p}~\frac{2t}{r_p^2}\exp
\left[-\lambda_{PT} \left(\pi r_I^2-2r_I^2 \arccos
\left(\frac{t}{2r_I}\right)+t\sqrt{r_I^2-\frac{t^2}{4}}\right)\right]\mathrm{d}t.
\end{eqnarray*}

By applying the inequality $\arccos (x)\leq \frac{\pi}{2}-x$ for
$0\leq x\leq 1$, we have
\begin{eqnarray*}
\mu &\leq&  \lambda_S \pi r_p^2 \int_0^{r_p}~\frac{2t}{r_p^2}\exp
\left\{-\lambda_{PT} \left[\pi r_I^2-2r_I^2
\left(\frac{\pi}{2}-\frac{t}{2r_I}\right)+t\sqrt{r_I^2-\frac{t^2}{4}}\right]\right\}\mathrm{d}t
\\
&\leq& \lambda_S \pi r_p^2 \int_0^{r_p}~\frac{2t}{r_p^2}\exp
(-\lambda_{PT} tr_I)\mathrm{d}t \\
&=& \lambda_S \pi \left(\frac{2}{\lambda_{PT}^2
r_I^2}-\frac{2}{\lambda_{PT}^2 r_I^2}\exp (-\lambda_{PT}\beta
r_I^2)-\frac{2\beta}{\lambda_{PT}}\exp (-\lambda_{PT}\beta
r_I^2)\right)\\
&\leq& \frac{2\lambda_S \pi}{\lambda_{PT}^2}(r_I)^{-2},
\end{eqnarray*}
where we have assumed that $r_p=\beta r_I$ $(0<\beta<1)$ under the
disk signal propagation and interference model. Since $r_I\propto
(p_{tx})^{1/\alpha}$, we arrive at Corollary~\ref{cor:cond_avg_deg}.


\bibliographystyle{ieeetr}

\end{document}